\documentclass[envcountsame,envcountsect,runningheads]{elsarticle}
\usepackage{latexsym}
\usepackage{amssymb}
\usepackage{epic}
\usepackage{graphicx}
\usepackage{tikz}
\usetikzlibrary{positioning,calc}

\usepackage{amsmath}
\usepackage{bbm}
\usepackage{mathrsfs}
\usepackage{caption}
\usepackage{cite}
\usepackage{wrapfig}

\usepackage{longtable}

\newtheorem{theorem}{Theorem}
\newtheorem{lemma}[theorem]{Lemma}
\newtheorem{corollary}[theorem]{Corollary}

\newenvironment{example}{\noindent{\bf Example:}}{}
\newenvironment{proof}{\noindent{\bf Proof:}}{
\par\vspace{3mm}}

\begin{document}

\title{Slowly synchronizing automata with fixed alphabet size}

\author[vu,ru]{Henk Don}
\ead{h.don@math.ru.nl}
\author[tue,ru]{Hans Zantema}
\ead{h.zantema@tue.nl}
\author[ru]{Michiel de Bondt}
\ead{m.debondt@math.ru.nl}

\address[vu]{ Free University Amsterdam, The Netherlands}
\address[ru]{ Radboud University Nijmegen, The Netherlands}
\address[tue]{Eindhoven University of Technology, The Netherlands}


\begin{abstract}
It was conjectured by \v{C}ern\'y in 1964 that a synchronizing DFA on $n$ states always has a
shortest synchronizing word of length at most $(n-1)^2$, and he gave a sequence of DFAs for which
this bound is reached.

In this paper, we investigate the role of the alphabet size. For each possible alphabet size, we count DFAs on $n \le 6$ states which synchronize in $(n-1)^2 - e$ steps, for all
$e < 2\lceil n/2 \rceil$. Furthermore, we give constructions of automata with any
number of states, and $3$, $4$, or $5$ symbols, which synchronize slowly, namely in
$n^2 - 3n + O(1)$ steps.

In addition, our results prove \v{C}ern\'y's conjecture for $n \le 6$. Our computation has led to $27$ DFAs on $3$, $4$, $5$ or $6$ states, which synchronize in $(n-1)^2$ steps,
but do not belong to \v{C}ern\'y's sequence. Of these $27$ DFA's, $19$ are new, and the
remaining $8$ which were already known are exactly the \emph{minimal} ones: they will not synchronize
any more after removing a symbol.

So the $19$ new DFAs are extensions of automata which were already known, including
the \v{C}ern\'y automaton on $3$ states. But for $n > 3$,
we prove that the \v{C}ern\'y automaton on $n$ states does not admit non-trivial
extensions with the same smallest synchronizing word length $(n-1)^2$.

%
\end{abstract}

\maketitle

\section{Introduction}

A {\em deterministic finite automaton (DFA)} over a finite alphabet $\Sigma$ is called
{\em synchronizing} if it admits a {\em synchronizing word}. Here a word $w \in \Sigma^*$ is
called {\em synchronizing} (or directed, or reset) if  starting in any state $q$, after
processing $w$ one always ends in one particular state $q_s$. So processing $w$ acts as a reset
button:
no matter in which state the system is, it always moves to the particular state $q_s$.
Now \v{C}ern\'y's conjecture
(\citep{C64}) states:
\begin{quote}
Every synchronizing DFA on $n$ states admits a synchronizing word of length $\leq (n-1)^2$.
\end{quote}

Surprisingly, despite extensive effort this conjecture is still open, and even the best known upper
bound is still cubic in $n$.  \v{C}ern\'y himself
(\citep{C64}) provided an upper bound of $2^n-n-1$ for the
length of the shortest synchronizing word. A substantial
improvement was given by Starke \citep{starke}, who was the first
to give a polynomial upper bound, namely
$1+\frac{1}{2}n(n-1)(n-2)$. The best known upper bound for a long times was $\frac{1}{6}(n^3-n)$, established by Pin in
1983 \citep{pin}. He reduced proving this upper bound to a purely
combinatorial problem which was then solved by Frankl \citep{frankl}.
Since then for more than 30 years only limited progress for the general case
has been made. Very recently a slight improvement was claimed by Szyku\l{}a \citep{szykula17}.

The conjecture has been proved for some particular classes of
automata, such as circular automata, aperiodic automata and
one-cluster automata with prime length cycle. For these results and some more partial
answers, see \citep{almeida,beal,don,dubuc,E90,steinberg}.
For a survey on synchronizing
automata and \v{C}ern\'y's conjecture, we refer to \citep{volkov}.

\begin{wrapfigure}{r}{4cm}

\vspace{-5mm}
\includegraphics[scale=0.25]{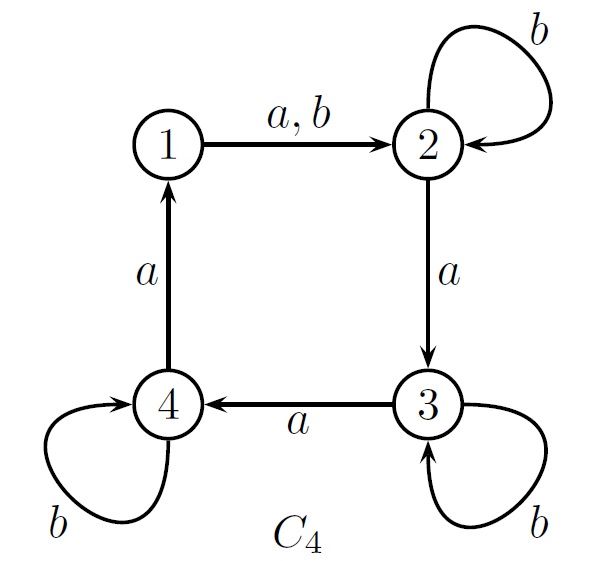}

\end{wrapfigure}

In \citep{C64}, \v{C}ern\'y already gave DFAs for which the bound of the conjecture is attained:
for $n \geq 2$ the DFA $C_n$ is defined
to consist of $n$ states $1,2,\ldots,n$, and two symbols $a,b$, acting by
$\delta(i,a) = i+1$ for $i = 1,\ldots,n-1$, $\delta(n,a) = 1$, and  $\delta(i,b) = i$ for
$i = 2,\ldots,n$, $\delta(1,b) = 2$. For $n=4$ this is depicted on the right.

For $C_n$ the string $w = b (a^{n-1}b)^{n-2}$ of length $|w| = (n-1)^2$ satisfies $qw = 2$ for all $q
\in Q$, so is synchronizing. No shorter synchronizing word exists for
$C_n$ as is shown in \citep{C64}, showing that the bound in \v{C}ern\'y's conjecture is sharp.

One topic of this paper is to investigate all DFAs for which the bound is reached; these
DFAs are called {\em critical}. Moreover, we also investigate bounds on synchronization lengths for
fixed alphabet size.  A DFA for which the bound $(n-1)^2$ is exceeded is called {\em
super-critical}, so \v{C}ern\'y's conjecture states that no super-critical DFA exists.
To exclude infinitely many trivial extensions, we only consider {\em basic} DFAs: no two distinct
symbols act in the same way in the automaton, and no symbol acts as the identity. Obviously, adding
the identity or copies of existing symbols has no influence on synchronization.

An extensive investigation was already done by Trahtman in \citep{T06}:
by computer support and clever algorithms
all critical DFAs on $n$ states and $k$ symbols were
investigated for $3 \leq n \leq 7$ and $k \leq 4$, and for $n=8,9,10$ and $k=2$. Here a
minimality requirement was added:  examples were excluded if criticality may be kept after
removing one symbol. Then up to isomorphism there are exactly 8 of them, apart
from the basic \v{C}ern\'y examples: 3 with 3 states, 3 with 4, one with 5 and one with
6. So apart from the basic \v{C}ern\'y examples only 8 other critical DFAs were known.
It was conjectured in \citep{T06} that no more exist, which is refuted in this paper by finding
several more not satisfying the minimality condition, all being extensions of known examples with 3 or 4 states.
As one main result we prove that up to isomorphism for $n=3$ there are exactly 15 basic
critical DFAs and for $n=4$ there are exactly 12 basic critical DFAs, 19 more than the four
for $n=3$ and the four for $n=4$ that were known before. For both $n=5$ and $n=6$ we prove that there
are no more basic critical DFAs than the two that were known before. For $n=3$ we give a self-contained
proof; for $n=4,5,6$ we exploit extensive computer support. For all $n \leq 6$ we investigate
the DFAs with several alphabet sizes and minimal synchronization lengths; as expected no super-critical DFAs
exist.

Two typical basic critical DFAs that were not known before are depicted as follows.

\begin{center}
\includegraphics[scale=0.25]{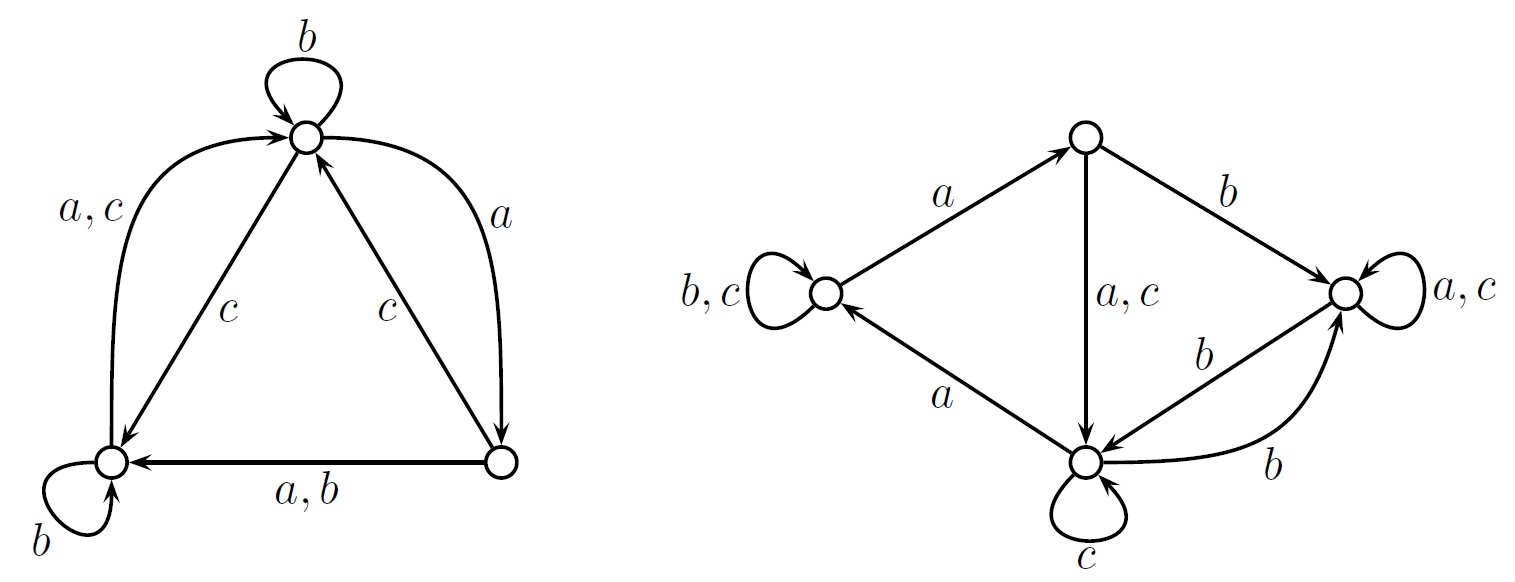}
\end{center}

The left one restricted to $a,b$ is exactly $C_3$, while restricted to $a,c$ it is exactly a DFA
found in \citep{T06} that we call T3-1 in Section \ref{sec34}. So this example is a kind of union
of $C_3$ and T3-1. It has four distinct synchronizing words of the minimal length 4 described
by $(b+c)aa(b+c)$, having two distinct synchronizing states.

The right one restricted to $a,b$ is the example found in \citep{CPR71}
that we call CPR in Section \ref{sec34}. However, the extra non-trivial symbol $c$ does not
occur in any known critical DFA on four states.
It has eight distinct synchronizing words of the minimal length 9 described by
$(b+c)aa(b+c)abaa(b+c)$, again having two distinct synchronizing states.

In the partial order on the 15 critical basic DFAs on three states, the four given in \citep{T06} are the
minimal ones, but there is only one maximal one, being an upper bound of all. Here {\em maximal} means that
that it does not admit an extension that is still basic and critical.
In the partial order on the 12 critical basic DFAs on four states, the four given in \citep{T06} are the
minimal ones, and exactly three are maximal. Two of the maximal examples are also minimal; the other is an
upper bound of the two remaining minimal ones.

For $n \geq 5$, we wonder whether the minimal critical DFAs in Trahtman's analysis admit critical
extensions just as for $n \leq 4$. The answer is negative. Apart from $C_n$ these include only two
minimal critical DFAs: one with 5 and one with 6 states, and our computer search shows that they do
not admit critical extensions. For $C_n$ this boils down to our
theorem stating that when adding an extra symbol to $C_n$ not acting as the identity or as
one of the existing symbols, always a strictly shorter synchronizing word can be obtained.
The theorem is proved by a case analysis in how this extra symbol acts on the states.

With two symbols the minimal synchronization
length $(n-1)^2$ can be reached for all $n$, but no critical DFAs with more than two symbols and at least $6$ states are known. It is a natural question which minimal
synchronization lengths can be reached with alphabet size $k > 2$. We prove that for $k = 3,4,5$ the
minimal synchronization length $n^2 - 3n +7 -k$ can be reached, and even for $k$ being exponential in $n$,
a quadratic expression in $n$ can be reached. 

This paper is mainly based on the LATA 2017 paper \citep{DZ17} by the first two authors, but also contains
several new contributions, in particular results for fixed alphabet size.
It is organized as follows. In Section \ref{secprel} we give some preliminaries. Section \ref{seclb} investigates
general lower bounds on synchronization length both depending on DFA size and alphabet size that were not published before.
Section \ref{sec34} investigates DFAs of at most six states. The resulting new critical DFAs on 3 and 4 states
already appeared in \citep{DZ17}, but here we extend the full analysis for 3 and 4 states to 5 and 6 states along
the lines of \citep{BDZ17}.
Moreover, here we do this not only for critical DFAs, but also for synchronization
length just below $(n-1)^2$, namely $(n-1)^2 - e$ steps for all $e < 2\lceil n/2 \rceil$,
and split up for distinct alphabet sizes.
A final part is Section \ref{seccn}, where we prove our property for $C_n$ for
arbitrary $n$: $C_n$ has no critical extension for $n \geq 5$.
This is done by an extensive case analysis showing that any extra non-trivial symbol $c$ acting on the $n$ states 
always yields a shorter synchronizing word. Here we give the full proof for which space was lacking in
\citep{DZ17}. We conclude in Section \ref{secconc}.

\section{Preliminaries}
\label{secprel}

A {\em deterministic finite automaton (DFA)} over a finite alphabet $\Sigma$ consists of a finite set
$Q$ of states and a map $\delta: Q \times \Sigma \to Q$.\footnote{For synchronization the
initial state and the set of final states in the standard definition may be ignored.}
A DFA is
called {\em basic} if the mappings $q \mapsto \delta(a,q)$ are distinct for all $a \in \Sigma$, and
are not the identity. For $w \in \Sigma^*$ and $q \in Q$ define $qw$ inductively by
$q \epsilon = q$ and $q w a = \delta(qw,a)$ for
$a \in \Sigma$. So  $qw$ is the state where one ends when starting in $q$ and applying
$\delta$-steps for the symbols in $w$ consecutively, and $qa$ is a short hand notation for $\delta(q,a)$.
A word $w \in \Sigma^*$ is called {\em synchronizing} if a
state $q_s \in Q$ exists such that $q w = q_s$ for all $q \in Q$. Stated in words: starting in
any state $q$, after processing $w$ one always ends in state $q_s$.
Obviously, if $w$ is a synchronizing word then so is $wu$ for any word $u$.
A DFA on $n$ states is
{\em critical} if its shortest synchronizing word has length $(n-1)^2$; it is
{\em super-critical} if its shortest synchronizing word has length $> (n-1)^2$.
A critical DFA is {\em minimal} if it is not the extension of another critical DFA by one or more
extra symbols; it is {\em maximal} if it does not admit a basic critical extension.

For $n\geq 2$ and $1\leq k\leq n^n-1$, we define $d(n,k)$ to be the maximal shortest synchronizing word
length in a synchronizing $n$-state basic DFA with alphabet size $k$.

The basic tool to analyze synchronization is by exploiting the {\em power set automaton}.
For any DFA $(Q,\Sigma, \delta)$ its power set automaton is the DFA $(2^Q,\Sigma, \delta')$
where $\delta' : 2^Q \times \Sigma \to 2^Q$ is defined by
$\delta'(V,a) = \{q \in Q \mid \exists p \in V : \delta(p,a) = q \}$.
For any $V \subseteq Q, w \in \Sigma^*$ we define $Vw$ as above, using $\delta'$ instead of
$\delta$.  From this definition one easily proves that
$Vw = \{ qw \mid q \in V \}$ for any $V \subseteq Q, w \in \Sigma^*$.
A set of the shape $\{q\}$ for $q \in Q$ is called a {\em singleton}.
So a word $w$ is synchronizing if and only if $Qw$ is a singleton.
Hence a DFA is synchronizing if and only if its
power set automaton admits a path from $Q$ to a singleton, and the shortest
length of such a path corresponds to the shortest length of a synchronizing word.

The power set automaton of $C_4$ is depicted below, in which indeed the unique shortest
path from $Q$ to a singleton (indicated by fat arrows from 1234 to 2) has length 9.

\begin{center}
\includegraphics[scale=0.25]{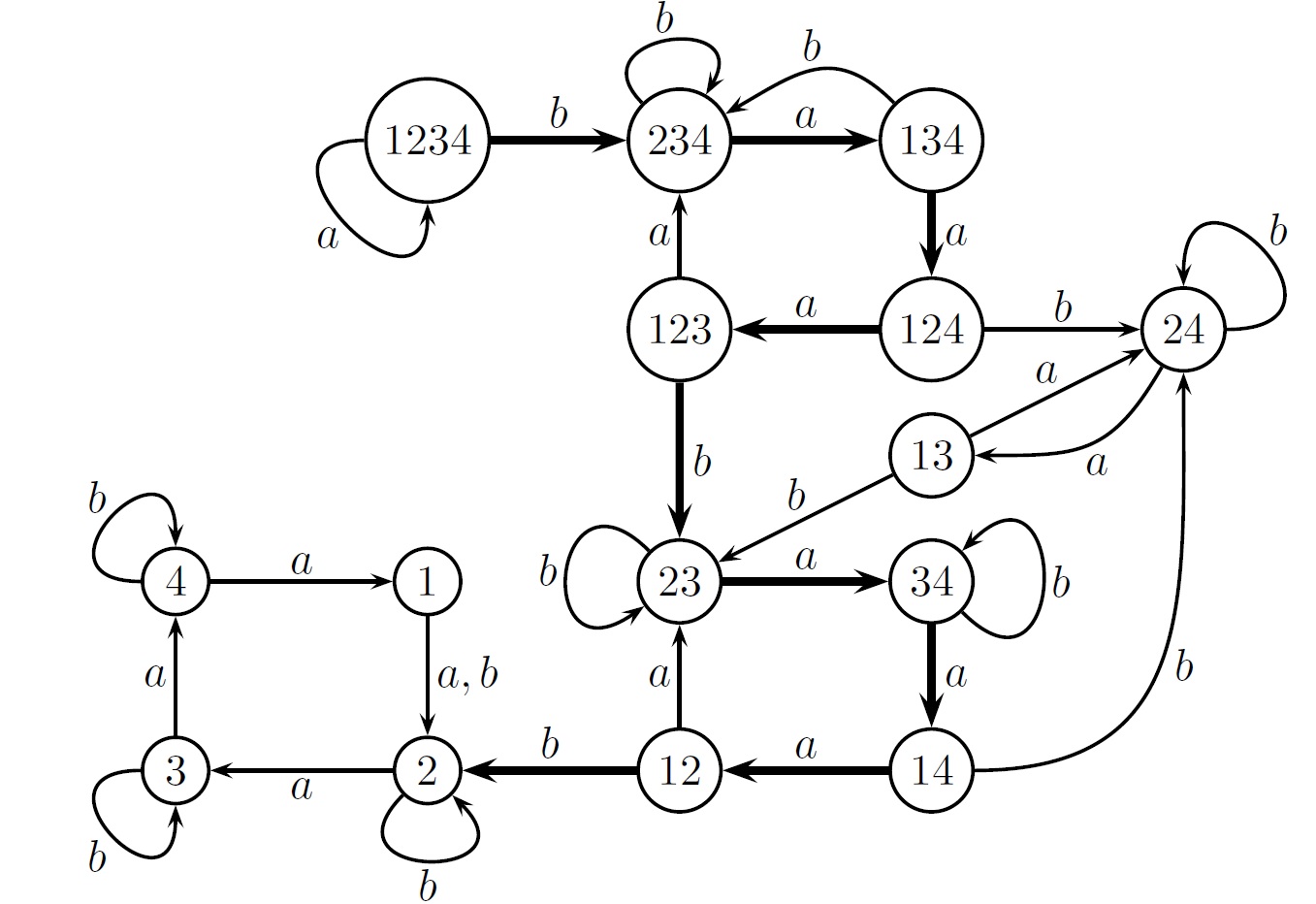}
\end{center}

\section{Lower bounds for fixed alphabet size}
\label{seclb}

A central question in this paper is how the maximal synchronizing word length of a DFA depends on the size of the alphabet.
The following theorem gives a quite straightforward construction to create DFAs with large alphabet and long shortest synchronizing words. This formalizes observations in
the same spirit that have been made before \citep{B14,AVG12}.
\begin{theorem}\label{theorem:dnk}
For $n\geq 2$, $0\leq m\leq n-2$ and $1\leq k \leq (n-m)^{n-m}-1$, the following inequality holds:
\[
d(n,(k+1)n^m-1)\geq d(n-m,k).
\]
\end{theorem}
\begin{proof}
Suppose that $A=(Q,\Sigma,\delta)$ is a basic DFA with $|Q|=n-m$, $|\Sigma|=k$ and shortest synchronizing word length $d:=d(n-m,k)$. Denote the states by $Q = \left\{1,2,\ldots,n-m\right\}$.

Now define a new DFA $B = (\tilde Q,\tilde \Sigma,\tilde \delta)$ with  $\tilde Q := \left\{1,2,\ldots,n\right\}$ as follows. To construct the alphabet $\tilde\Sigma$, first add the identity symbol to $\Sigma$ and let this set be called $\Sigma^+$. For each $a\in\Sigma^+$, define all $n^m$ possible symbols on $\tilde Q$ that coincide with $a$ when restricted to $Q$. Removing the identity symbol on $\tilde Q$, this gives the alphabet $\tilde\Sigma$ containing $(k+1)n^m-1$ symbols on $\tilde Q$.

When restricted to $Q$, the automata $A$ and $B$ have the same symbols. Every synchronizing word $w$ for $B$ corresponds to a synchronizing word for $A$, and therefore has length at least $d$. Since $\tilde\Sigma$ contains all possible extensions of letters in $\Sigma$, in particular there exists a letter $a\in\tilde \Sigma$ for which $\tilde Qa\subseteq Q$, proving that $B$ is synchronizing.
\qed\end{proof}

\begin{corollary}\label{cor:large_alphabet}
For $n\geq 2$ and $2\leq k \leq 3n^{n-2}-1$,
\[
d(n,k) \geq \left(n-\left\lceil\frac{\log(k+1)-\log(3)}{\log(n)}\right\rceil-1\right)^2.
\]
\end{corollary}
\begin{proof}
First note that $d(n,2)\geq (n-1)^2$, because \v{C}ern\'y's automata $C_n$ have two symbols and attain this length. Now let $n$ and $k$ be as in the corollary and let \[m=\left\lceil\frac{\log(k+1)-\log(3)}{\log(n)}\right\rceil.\]
Then $m\leq n-2$, so the automaton $C_{n-m}$ is well-defined, has two symbols and shortest synchronizing word length $(n-m-1)^2$. Theorem \ref{theorem:dnk} now gives
\[
d(n,3n^m-1)\geq d(n-m,2)\geq (n-m-1)^2.
\]
So there exists a synchronizing automaton on $n$ states with $3n^m-1\geq k$ symbols, and shortest synchronizing word length at least $(n-m-1)^2$. Furthermore, only two of these symbols are needed for synchronization. Therefore, we can just remove some symbols to reach the required alphabet size $k$, and the conclusion follows.
\qed\end{proof}

The corollary gives an easy lower bound for $d(n,k)$, which is however not sharp. For instance, for $k=3$, we find $d(n,3)\geq (n-2)^2$, which is obtained by taking $C_{n-1}$ and adding one extra state as in the proof of Theorem \ref{theorem:dnk}. It is easy to see that in fact $d(n,3)\geq (n-2)^2+1$, since the extra state can be used to extend the shortest synchronizing word by at least one letter. In the remainder of this section, we focus on more substantial improvements for small values of $k$. The following result gives further improvements for alphabet size 3, 4 and 5. Sequences of automata with these synchronization lengths were already known, but they all had only two symbols \citep{AVZY07,AVG12}.

\begin{theorem}\label{theorem:alphabet345}
For $n\geq 3$ and alphabet size $k=3, 4$ and $5$, the maximal shortest synchronizing word length $d(n,k)$ satisfies
\begin{align*}
d(n,3) &\geq n^2-3n+4,\\
d(n,4) &\geq n^2-3n+3,\\
d(n,5) &\geq n^2-3n+2.
\end{align*}
\end{theorem}

\begin{proof}
Consider the basic automaton $A$ with state set $Q = \left\{1,\ldots,n\right\}$ alphabet $\Sigma = \left\{a,b,c,d,e\right\}$, where $qx$ for each state $q$ and symbol $x$ is defined as follows:
\[
\renewcommand{\arraystretch}{1.4}
\begin{array}{|@{\quad}c@{\quad}|@{\quad}c@{\quad}c@{\quad}c@{\quad}c@{\quad}|}
\hline
 & q=1 & q=2 & q=3 & q \geq 4 \\
\hline
x=a & 2 & 3 & 4 & (q+1) \bmod n \\[-5pt]
x=b & 1 & 3 & 3 & q \\[-5pt]
x=c & 3 & 3 & 4 & (q+1) \bmod n \\[-5pt]
x=d & 2 & 4 & 4 & (q+1) \bmod n \\[-5pt]
x=e & 3 & 4 & 4 & (q+1) \bmod n \\
\hline
\end{array}
\]
The automaton $A$ is depicted below. We will prove the following three claims:
\begin{enumerate}
	\item $A$ has shortest synchronizing word length $n^2-3n+2$,
	\item $A^{-d}:=\left\{Q,\left\{a,b,c,e\right\},\delta\right\}$ has shortest synchronizing word length $n^2-3n+3$,
	\item $A^{-cd}:=\left\{Q,\left\{a,b,e\right\},\delta\right\}$ has shortest synchronizing word length $n^2-3n+4$.
\end{enumerate}

\begin{tikzpicture}[-latex ,node distance =3 cm and 3cm ,on grid ,
semithick ,
state/.style ={ circle ,top color =white , bottom color = white!20 ,
	draw, black , text=black , minimum width =.7 cm}]

\node[state] (n) {$n$};
\node[state] (1) at ($ (n) +  (60:2) $) {1};
\node[state] (2) at ($ (1) +  (30:2.5) $) {2};
\node[state] (3) at ($ (2) +  (0:2.5) $) {3};
\node[state] (4) at ($ (3) +  (-30:2.5) $) {4};
\node[state] (5) at ($ (4) +  (-60:2) $) {5};

\node[state,draw, white , text=white] (6) at ($ (5) +  (-90:2) $) {};
\node[state,draw, white , text=white] (7) at ($ (n) +  (-90:2) $) {};
\node[state,draw, white , text=black] (dummy) at ($ (2) +  (0:1.25) $) {};
\node[state,draw, white , text=black] (dummy2) at ($ (dummy) +  (-90:4.5) $) {\ldots};

\path (1) edge node[above left] {$a,d$} (2);
\path (2) edge node[above] {$a,b,c$} (3);
\path (3) edge node[above right] {$a,c,d,e$} (4);
\path (4) edge node[right] {$a,c,d,e$} (5);
\path (5) edge node[left] {$a,c,d,e$} (6);
\path (7) edge node[right] {$a,c,d,e$} (n);
\path (n) edge node[left] {$a,c,d,e$} (1);

\path (1) edge node[below] {$c,e$} (3);
\path (2) edge node[below] {$d,e$} (4);

\path (n) edge[in=140,out=190,loop] node[left]{$b$} (n);
\path (1) edge[in=110,out=160,loop] node[left]{$b$} (1);
\path (3) edge[in=50,out=100,loop]node[above] {$b$} (3);
\path (4) edge[in=20,out=70,loop] node[right]{$b$} (4);
\path (5) edge[in=-10,out=40,loop] node[right]{$b$} (5);
\end{tikzpicture}

When restricted to the symbols $a$ and $b$, $A$ is equal to \v{C}ern\'y's automaton $C_n$, so $A$ is synchronizing. Let $w$ be a shortest synchronizing word for $A$. We will show that there exists a shortest synchronizing word containing only the letters $c$ and $d$.

Note that for all states $q$, the following is true:
\[
qab = qc,\quad qbb=qb,\quad qcb=qc,\quad qdb = qbc \quad\text{and}\quad qeb = qbc.
\]
Therefore, if $w_iw_j$ is a factor of $w$ and $w_j=b$, then $w_i\in\left\{d,e\right\}$ and we can replace $w_iw_j$ by the string $bc$. Doing this repeatedly, we can assume that $b$ does not occur in $w$ (if the first letter of $w$ is a $b$, it can be replaced by $c$ since $Qb=Qc$).

If $S\subseteq Q$ and $1\not\in S$, then $Sa = Sc$ and $Se=Sd$. If $S\subseteq Q$ and $2\not\in S$, then $Sa = Sd$ and $Se=Sc$. If $S\subseteq Q$ and $\left\{1,2\right\}\subseteq S$, then $Sc\subset Sa$ and $Sc\subseteq Se$. These observations show that every occurrence of $a$ or $e$ in $w$ can be replaced by $c$ or $d$. Thus, the restriction $A_{cd}$ of $A$ to the symbols $c$ and $d$ has the same synchronizing word length as $A$ itself. This restriction was shown to have shortest synchronizing word length $n^2-3n+2$ in \citep{AVG12}, establishing the lower bound $d(n,5)\geq n^2-3n+2$.

To prove the second claim, suppose $w=w_1w_2\ldots$ is synchronizing for $A$ with $w_1\in\left\{b,c,e\right\}$. Then there exists a shorter synchronizing word: if $w_2\in\left\{b,c,e\right\}$, then $Qw_1w_2=Qw_1$ and if $w_2\in\left\{a,d\right\}$, then $Qw_1w_2=Qd$. Clearly a shortest synchronizing word also does not start with $a$ because $Qa=Q$. Therefore, every shortest synchronizing word for $A$ starts with the symbol $d$, proving that the shortest synchronizing word for $A^{-d}$ has length at least $n^2-3n+3$.

The word $(ba^{n-1})^{n-2}b$ is synchronizing for $C_n$ and therefore also for $A^{-d}$. Replacing all occurrences of $ab$ by $c$ proves that $b(a^{n-2}c)^{n-2}$ is synchronizing for $A^{-d}$, showing that the shortest synchronizing word for $A^{-d}$ has length at most $n^2-3n+3$. Hence, $d(n,4)\geq n^2-3n+3$.

Assume that $w=w_1\ldots w_{m-1}w_m$ is a shortest synchronizing word for $A$. This means that $|Qw_1\ldots w_{m-1}|\geq 2$ and $|Qw_1\ldots w_{m}|=1$. Obviously, $w_m\neq a$. If $w_m\in\left\{b,d,e\right\}$, then the only possibility is $Qw_1\ldots w_{m-1}=\left\{2,3\right\}$ and consequently $w_{m-1}=a$. It follows that in this case $Qw_1\ldots w_{m-2} = \left\{1,2\right\}$. But then $w_1\ldots w_{m-2}c$ is synchronizing as well, contradicting the assumption. Therefore, every shortest synchronizing word for $A$ ends with the symbol $c$, proving that the shortest synchronizing word for $A^{-cd}$ has length at least $n^2-3n+4$.

We will finish the proof by constructing a synchronizing word of this length. Observe that $1ea^{n-2}=1$, $2ea^{n-2}=2$ and $qea^{n-2}=q-1$ for $3\leq q\leq n$. This implies that $Q(ea^{n-2})^{n-2}= \left\{1,2\right\}$. Since $\left\{1,2\right\}ae = \left\{4\right\}$, it follows that $w = (ea^{n-2})^{n-2}ae$ is synchronizing and it has length $n^2-3n+4$.
\qed\end{proof}

As the shortest synchronizing word for $A^{-cd}$ only contains the symbols $a$ and $e$, the automaton
$A^{-bcd} := \left\{Q,\left\{a,e\right\},\delta\right\}$ has shortest synchronizing word length
$n^2-3n+4$ as well. From the proof it also follows that $A^{-c}:= \left\{Q,\left\{a,b,d,e\right\},\delta\right\}$ has shortest synchronizing word length $n^2-3n+3$.

We note here that the automata constructed in Theorem \ref{theorem:alphabet345} are extensions of
\v{C}ern\'y's automata $C_n$ with synchronizing length close to $(n-1)^2$. In Section \ref{seccn},
we will show that all possible extensions of $C_n$ have synchronizing length strictly less than $(n-1)^2$.

Experimental results reported in \citep{T06} and \citep{AVG12} give evidence that 
for $n \geq 7$, no two-letter automata exist with synchronization length strictly between $n^2-3n+4$ and $(n-1)^2$. Moreover, no sequences of automata are known with synchronization length strictly between $n^2-4n+7$ and $n^2-3n+2$. The  synchronization lengths between these two gaps correspond exactly to our lower bounds for alphabet size 3, 4 and 5. Moreover, for $n = 7,8,9,10$, the value $n^2-3n+4$ also matches the value just below the gap experimentally found by Trahtman. An open question is if sequences of automata with synchronization length above $n^2-4n+O(1)$ and at least $6$ symbols exist.

The constructions of Theorem \ref{theorem:alphabet345} can be used to improve on the bound of
Corollary \ref{cor:large_alphabet} for certain alphabet sizes.
For example, if $3n\leq k\leq 6n-1$, Corollary \ref{cor:large_alphabet} gives
$d(n,k)\geq (n-3)^2 = n^2-6n+9$. However, when we apply Theorem \ref{theorem:dnk}
with $k=5$ and $m=1$, we obtain
\[
d(n,6n-1) \geq d(n-1,5) = n^2-5n+6.
\]
The same lower bound also applies to other alphabet sizes $k$ for which $3n\leq k\leq 6n-1$.

\section{Small DFAs}
\label{sec34}

In this section we exploit computer support to investigate all DFAs on $n=2,3,4,5,6$ states having long shortest synchronization length.
As the number of DFAs on $n$ states grows like $2^{n^{n}}$, an exhaustive search is a non-trivial affair,
even for small values of $n$. The problem is that the alphabet size in a basic DFA can be as large as
$n^n-1$, while in earlier work only DFAs with at most four symbols were checked by Trahtman \citep{T06}.
In \citep{DZ17} we gave a full investigation of all critical DFAs on 3 and 4 states, without restriction on
the alphabet size. In \citep{BDZ17} we extended this to 5 and 6 states. In this paper we extend this work further by
not restricting to critical DFAs, that is, having synchronization length $(n-1)^2$, but investigate longest possible
synchronization lengths not only depending on the number $n$ of states, but also depending on the alphabet size $k$.
Before giving the results first we explain the underlying ideas of our algorithm;
following the same lines as in \citep{BDZ17}. We use the following terminology. A DFA $\mathcal{B}$ obtained by adding
some symbols to a DFA $\mathcal{A}$ will be called an \emph{extension} of $\mathcal{A}$. If $\mathcal{A} =
(Q,\Sigma,\delta)$, then $S\subseteq Q$ will be called \emph{reachable} if there exists a word
$w\in\Sigma^*$ such that $Qw=S$. We say
that $S$ is \emph{reducible} if there exists a word $w$ such that $|Sw|<|S|$, and we call $w$ a
\emph{reduction word} for $S$. Our algorithm is mainly based on the following immediate observation:
\begin{lemma}\label{property:extension}
If a DFA $\mathcal{A}$ is synchronizing, and $\mathcal{B}$ is an extension of $\mathcal{A}$, then
$\mathcal {B}$ is synchronizing as well and its shortest synchronizing word is at most as long as the
shortest synchronizing word for $\mathcal{A}$.
\end{lemma}

The algorithm roughly runs as follows. We search for DFAs on $n$ states with synchronization length $s$, so a DFA is
discarded if it synchronizes faster, or if it does not synchronize at all. For a given DFA
$\mathcal{A} = (Q,\Sigma,\delta)$ which is not yet discarded or investigated, the algorithm does the
following:
\begin{enumerate}
\item If $\mathcal{A}$ is synchronizing with synchronization length $s$, we have identified an example we are
searching for.
\item If $\mathcal{A}$ is synchronizing with synchronization length $< s$, it is discarded, together with all its
possible extensions (justified by Lemma \ref{property:extension}).
\item If $\mathcal{A}$ is not synchronizing, then find an upper bound $L$ for how fast any synchronizing
extension of $\mathcal{A}$ will synchronize (see below). If $L < (n-1)^2$, then discard $\mathcal{A}$
and all its extensions. Otherwise, discard only $\mathcal{A}$ itself.
\end{enumerate}

The upper bound $L$ for how fast any synchronizing extension of $\mathcal{A}$ will synchronize, is
found by analyzing distances in the directed graph of the power automaton of $\mathcal{A}$. For
$S,T\subseteq Q$, the distance from $S$ to $T$ in this graph is equal to the length of the shortest
word $w$ for which $Sw=T$, if such a word exists.
We compute $L$ as follows:
\begin{enumerate}
\item Determine the size $|S|$ of a smallest reachable set. Let $m$ be the minimal distance from $Q$
to a set of size $|S|$.
\item For each $k\leq |S|$, partition the collection of irreducible sets of size $k$ into strongly
connected components. Let $m_k$ be the number of components plus the sum of their diameters.
\item For each reducible set of size $k\leq |S|$, find the length of its shortest
reduction word. Let $l_k$ be the maximum of these lengths.
\item Now note that a synchronizing extension of $\mathcal{A}$ will have a synchronizing
word of length at most
\[ L \; = \; m+\sum_{k=2}^{|S|}(m_k+l_k).  \]
\end{enumerate}

The algorithm performs a depth-first search. So after investigating a DFA, first all its extensions (not
yet considered) are investigated before moving on.
Still, we can choose which extension to pick first. We would like to choose
an extension that is likely to be discarded immediately together with all its extensions. Therefore,
we apply the following heuristic: for each possible extension $\mathcal{B}$ by one symbol, we count how
many pairs of states in $\mathcal{B}$ would be reducible. The extension for which this is maximal is
investigated first. The motivation is that a DFA is synchronizing if and only if each pair is reducible \citep{C64}.

Finally, we note that we have described a primitive version of the algorithm here. The algorithm
which has actually been used also takes symmetries into account, making it almost $n!$ times faster.
For the source code, we refer to \citep{BDZ17a}.

In the rest of this section we consecutively present the results for $n=2,3,4,5,6$ states. For every of these number of states
we explicitly present all basic critical DFAs, and give a table of the number of basic DFAs of synchronization length $s$ for
various values of $s$, all up to symmetry.

\subsection{Two States}
The case for two states is quite degenerate as every synchronizing DFA synchronizes by a single symbol, but for completeness we include it.
If the two states are 1, 2, then there are three possible symbols that are not the identity: $a$ mapping both states to 1, $b$ mapping both states to 2,
and $c$ swapping the two states. Every non-empty set of these symbols yields a synchronizing DFA, except for $\{c\}$. So with one symbol we have the two
symmetrical cases $\{a\}$ and $\{b\}$, yielding one up to symmetry. For two symbols we have the two symmetrical cases $\{a,c\}$ and $\{b,c\}$, and
$\{a,b\}$, yielding two up to symmetry. And finally we have $\{a,b,c\}$ with three symbols. As we will do for all $n=2,3,4,5,6$, we put these results in a table,
counting for every number of symbols (alph. size) and every minimal synchronization length (sync.) the number of corresponding basic DFAs, up to symmetry.
So for $n=2$ this yields the following table:

\begin{center}
\renewcommand{\arraystretch}{1.4}
\begin{tabular}{|r|r|}
\hline
alph.\@ & sync.\@ \\[-5pt]
size & 1 \\
\hline
1 & 1 \\[-5pt]
2 & 2 \\[-5pt]
3 & 1 \\
\hline
total & 4 \\
\hline
\end{tabular}
\end{center}

\subsection{Three States}

For three states our algorithm yields the following table for all synchronizing basic DFAs:

\begin{center}
\renewcommand{\arraystretch}{1.4}
\begin{tabular}{|r|rrrr|}
\hline
alph.\@ & sync.\@ & sync.\@ & sync.\@ & sync.\@ \\[-5pt]
size & 4 & 3 & 2 & 1 \\
\hline
1 & & & 1 & 1 \\[-5pt]
2 & 2 & 4 & 28 & 13 \\[-5pt]
3 & 7 & 32 & 249 & 145 \\[-5pt]
4 & 5 & 85 & 1410 & 1028 \\[-5pt]
5 & 1 & 107 & 5527 & 5394 \\[-5pt]
6 & & 81 & 16833 & 21610 \\[-5pt]
7 & & 39 & 40917 & 68916 \\[-5pt]
8 & & 10 & 81881 & 178855 \\[-5pt]
9 & & 2 & 136373 & 384897 \\[-5pt]
10 & & & 190932 & 695038 \\[-5pt]
11 & & & 225589 & 1062915 \\[-5pt]
12 & & & 225589 & 1384909 \\[-5pt]
13 & & & 190932 & 1543472 \\[-5pt]
14 & & & 136375 & 1474123 \\[-5pt]
15 & & & 81891 & 1206613 \\[-5pt]
16 & & & 40956 & 845014 \\[-5pt]
17 & & & 16914 & 504358 \\[-5pt]
18 & & & 5638 & 255108 \\[-5pt]
19 & & & 1508 & 108364 \\[-5pt]
20 & & & 303 & 38221 \\[-5pt]
21 & & & 48 & 10984 \\[-5pt]
22 & & & 5 & 2531 \\[-5pt]
23 & & & 1 & 447 \\[-5pt]
24 & & & & 61 \\[-5pt]
25 & & & & 6 \\[-5pt]
26 & & & & 1 \\
\hline
total & 15 & 360 & 1399900 & 9793024 \\
\hline
\end{tabular}
\end{center}

So no super-critical DFAs exist (sync. $> 4$), and there are exactly 15 basic critical DFAs. These will be investigated in Theorem \ref{thm3st}.
Before we do so, we recall the minimal critical DFAs as presented in \citep{T06} on three states, apart from $C_3$:

\begin{center}
\includegraphics[scale=0.25]{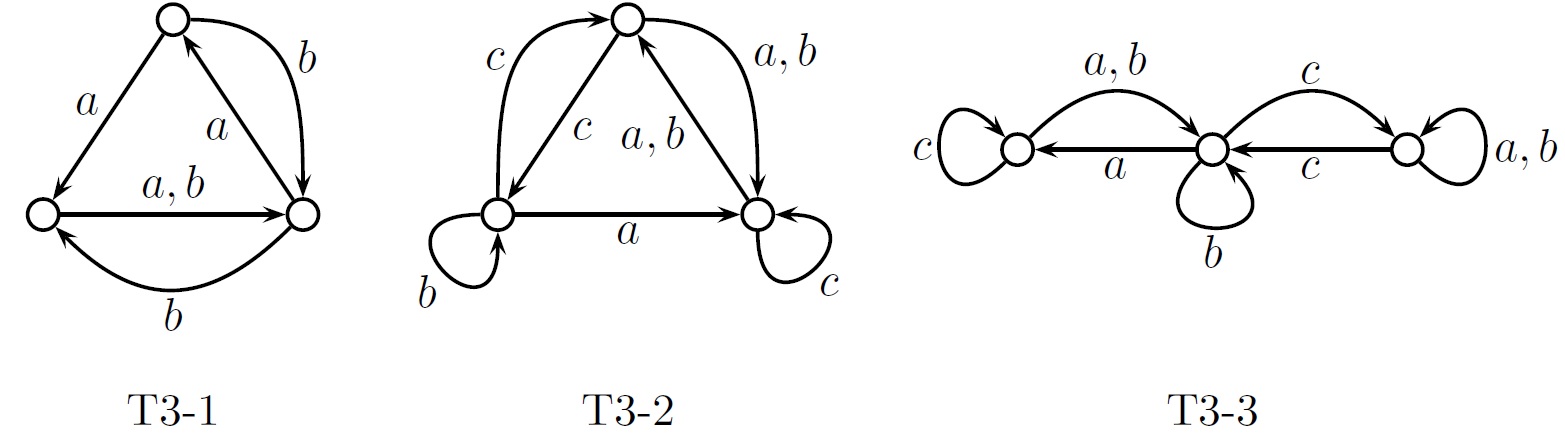}
\end{center}

We call them T3-1, T3-2 and T3-3, as they were found by Trahtman.
They all have a unique synchronizing word of length 4, being $baab$, $acba$, $bacb$,
respectively.

They can be combined to a single DFA A3 on five symbols $a,b,c,d,e$, depicted as follows.

\begin{minipage}{6cm}
\includegraphics[scale=0.24]{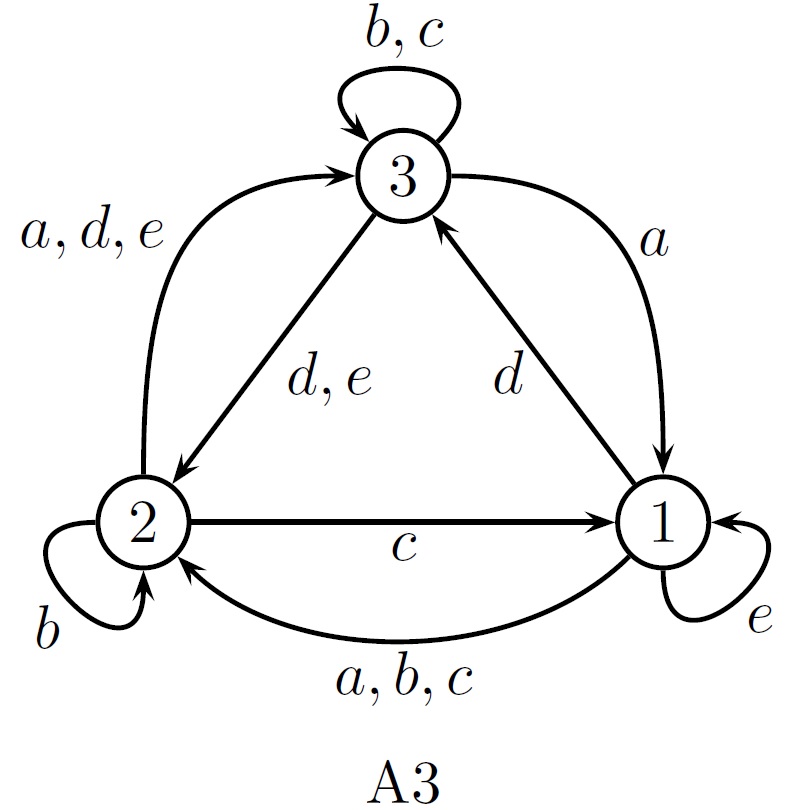}
\end{minipage}
\begin{minipage}{6cm}
\includegraphics[scale=0.24]{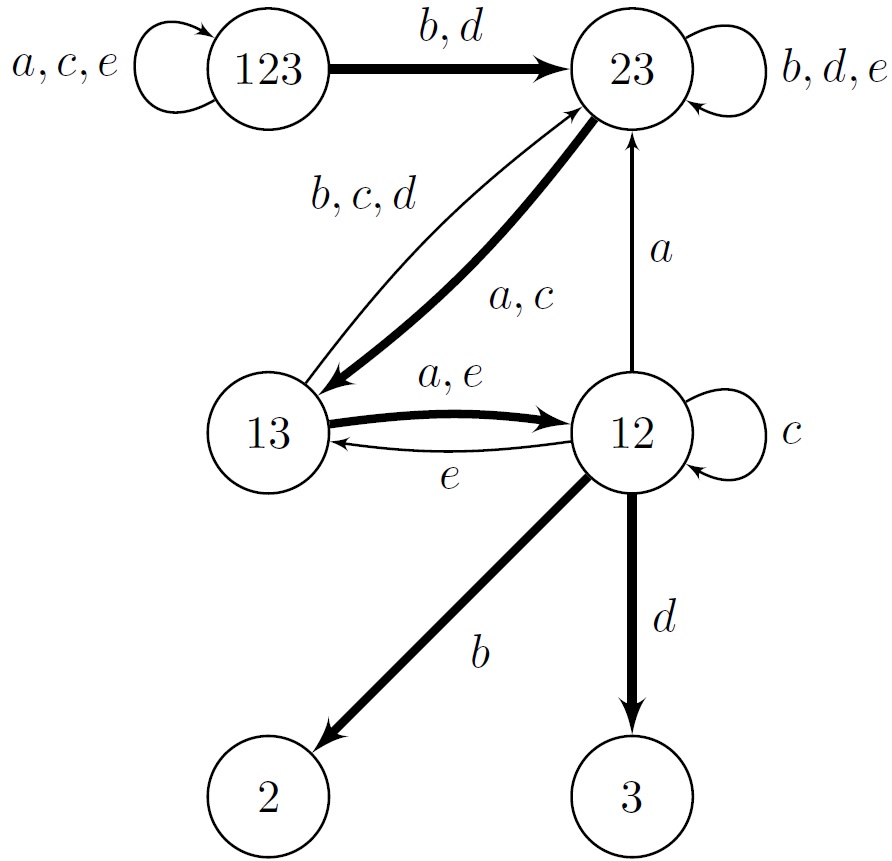}
\end{minipage}

\vspace{3mm}

Observe that A3 restricted to $a,b$ coincides with $C_3$,
A3 restricted to $a,d$ coincides with T3-1,
A3 restricted to $c,d,e$ coincides with T3-2 and
A3 restricted to $b,c,e$ coincides with T3-3, so exactly the four minimal critical automata on
three states from \citep{T06}.
On the other hand, as all minimal basic critical DFAs on three states are contained in
A3, A3 is the only maximal  basic critical DFA on three states. It admits 16 synchronizing words of
length 4, expressed by the regular expression $(b+d)(a+c)(a+e)(b+d)$, where state 2 is the
synchronizing state if the word ends in $b$ and state 3 if the word ends in $d$.

This follows from the analysis of the power set automaton of A3 as depicted right from A3 itself (we stopped when a singleton was reached).
Here the shortest paths from 123 to a singleton are indicated by fat arrows.

The relationship between A3 and critical DFAs is given in the following theorem.

\begin{theorem}
\label{thm3st}
No super-critical DFAs on three states exist, and a basic DFA on three states is critical if and only if
up to isomorphism it is one of the
15 automata that can be obtained from A3 by removing zero or more symbols and
keeping at least one of the sets $\{a,b\}, \{a,d\}, \{b,c,e\}, \{c,d,e\}$ of symbols.
\end{theorem}

This theorem follows from the results of our algorithm, but now we also give a self-contained proof that does not require
computer support.

\begin{proof}
Let $1,2,3$ be the three states. The automaton has a shortest synchronizing word of length $\geq
4$ if and only if the shortest path from $\{1,2,3\}$ to a singleton in the power set automaton has
length $\geq 4$.
There is a step from $\{1,2,3\}$ to a smaller set. Since the length of the shortest path is
$\geq 4$, this smaller set
is not a singleton, so it is a pair; without loss of generality we may assume this is $\{2,3\}$.

Let
$b$ be the first symbol of a shortest synchronizing word, so $\{1,2,3\} \stackrel{b}{\to}
\{2,3\}$. Since the shortest path from $\{2,3\}$ to a singleton consists of at least three steps, it
meets the other two pairs and consists of exactly three steps, yielding shortest synchronizing
word length 4. Maybe after swapping 2 and 3 we may assume this shortest path is
$\{1,2,3\} \stackrel{b}{\to} \{2,3\} \to \{1,3\} \to \{1,2\} \to \mbox{singleton}$.
As it is the shortest path, we conclude that for every symbol $a$ we have
\begin{enumerate}
\item either $\{1,2,3\} \stackrel{a}{\to} \{1,2,3\}$ or $\{1,2,3\} \stackrel{a}{\to} \{2,3\}$,
\item either $\{2,3\} \stackrel{a}{\to} \{2,3\}$ or $\{2,3\} \stackrel{a}{\to} \{1,3\}$, and
\item not $\{1,3\} \stackrel{a}{\to}$ singleton.
\end{enumerate}
We distinguish the cases of 2\@. Suppose first that $\{2,3\} \stackrel{a}{\to} \{2,3\}$.
Then 1.\@ becomes void, and for the images of $2$ and $3$, there are two options. For the image
of $1$, there are two options as well, so there are $4$ symbols, which are the symbols
$b$, $d$ and $e$ in A3, and the identity.
Suppose next that $\{2,3\} \stackrel{a}{\to} \{1,3\}$. Then 1.\@ can only be met if
$a$ permutes $\{1,2,3\}$, so there are two options, corresponding to the symbols
$a$ and $c$ in A3.

So for all DFAs being a sub-automaton of A3 it holds that if it is synchronizing, then the
shortest synchronizing word length is 4.
Restricting A3 to either $\{a,b\}$, $\{a,d\}$, $\{b,c,e\}$ or $\{c,d,e\}$ yields one of the
known synchronizing DFAs, so every extension is synchronizing too. Conversely, it is easily
checked that all of these restrictions are minimal: all symbols are required for
synchronization. This concludes the proof.
\qed
\end{proof}

As a consequence of Theorem \ref{thm3st} apart from the four minimal critical DFAs that were
known on three states, we obtain 11 more that are not minimal.

\subsection{Four States}
For four states our algorithm yields the following table for all synchronizing basic DFAs with minimal synchronization
length $\geq 6$:

\begin{center}
\renewcommand{\arraystretch}{1.4}
\begin{tabular}{|r|rrrr|}
\hline
alph.\@ & sync.\@ & sync.\@ & sync.\@ & sync.\@ \\[-5pt]
size & 9 & 8 & 7 & 6 \\
\hline
1 & & & & \\[-5pt]
2 & 2 & 5 & 11 & 21 \\[-5pt]
3 & 5 & 57 & 187 & 641 \\[-5pt]
4 & 4 & 146 & 979 & 5585 \\[-5pt]
5 & 1 & 151 & 2866 & 25538 \\[-5pt]
6 & & 72 & 5974 & 75372 \\[-5pt]
7 & & 15 & 9580 & 157414 \\[-5pt]
8 & & 1 & 12136 & 243850 \\[-5pt]
9 & & & 12239 & 287208 \\[-5pt]
10 & & & 9838 & 260468 \\[-5pt]
11 & & & 6286 & 182453 \\[-5pt]
12 & & & 3162 & 98120 \\[-5pt]
13 & & & 1230 & 39867 \\[-5pt]
14 & & & 360 & 11851 \\[-5pt]
15 & & & 76 & 2444 \\[-5pt]
16 & & & 11 & 312 \\[-5pt]
17 & & & 1 & 20 \\
\hline
total & 12 & 447 & 64936 & 1391164 \\
\hline
\end{tabular}
\end{center}

For alphabet size $\geq 18$ there are no DFAs with shortest synchronization length $\geq 6$.

In order to investigate all 12 (basic) critical DFAs on four states, first we give the minimal
critical DFAs as presented in \citep{T06} on four states, apart from $C_4$.
\begin{wrapfigure}{r}{5.5cm}

\includegraphics[scale=0.25]{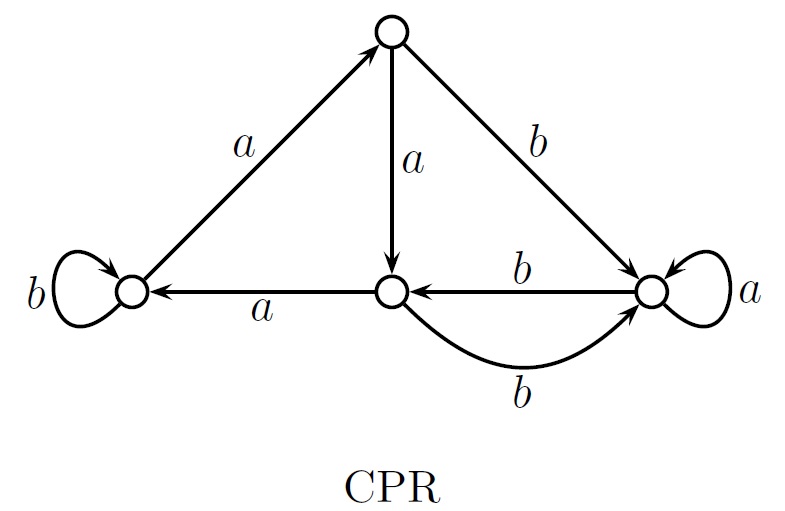}

\end{wrapfigure}
The first one is CPR, found by \v{C}ern\'y, Piricka and
Rosenauerova, \citep{CPR71}, and has unique synchronizing word of length 9, being $baababaab$.
The next two we call T4-1 and T4-2, as they were found by Trahtman.
The DFA T4-1 has a unique synchronizing word of length 9, being $abcacabca$; for T4-2 there are 4
synchronizing words of length 9 represented by $acb(a+c)a(a+b)cba$.

\begin{center}
\includegraphics[scale=0.25]{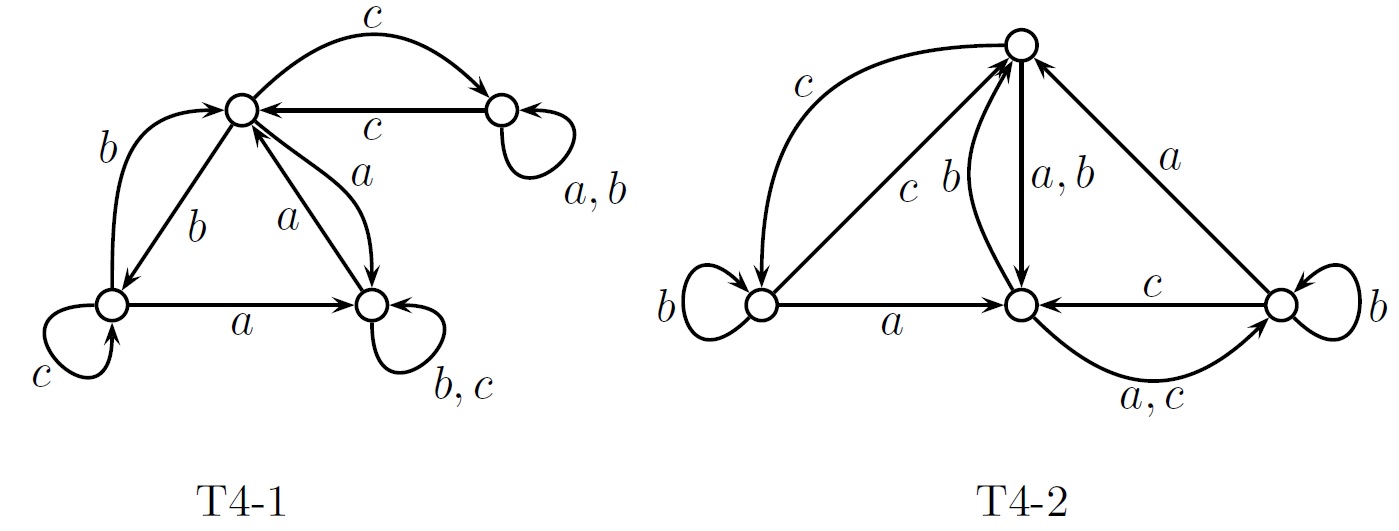}
\end{center}

In order to investigate all critical DFAs with four states, we introduce the DFA
A4 on five symbols $a,b,c,d,e$, depicted as follows.

\begin{center}
\includegraphics[scale=0.25]{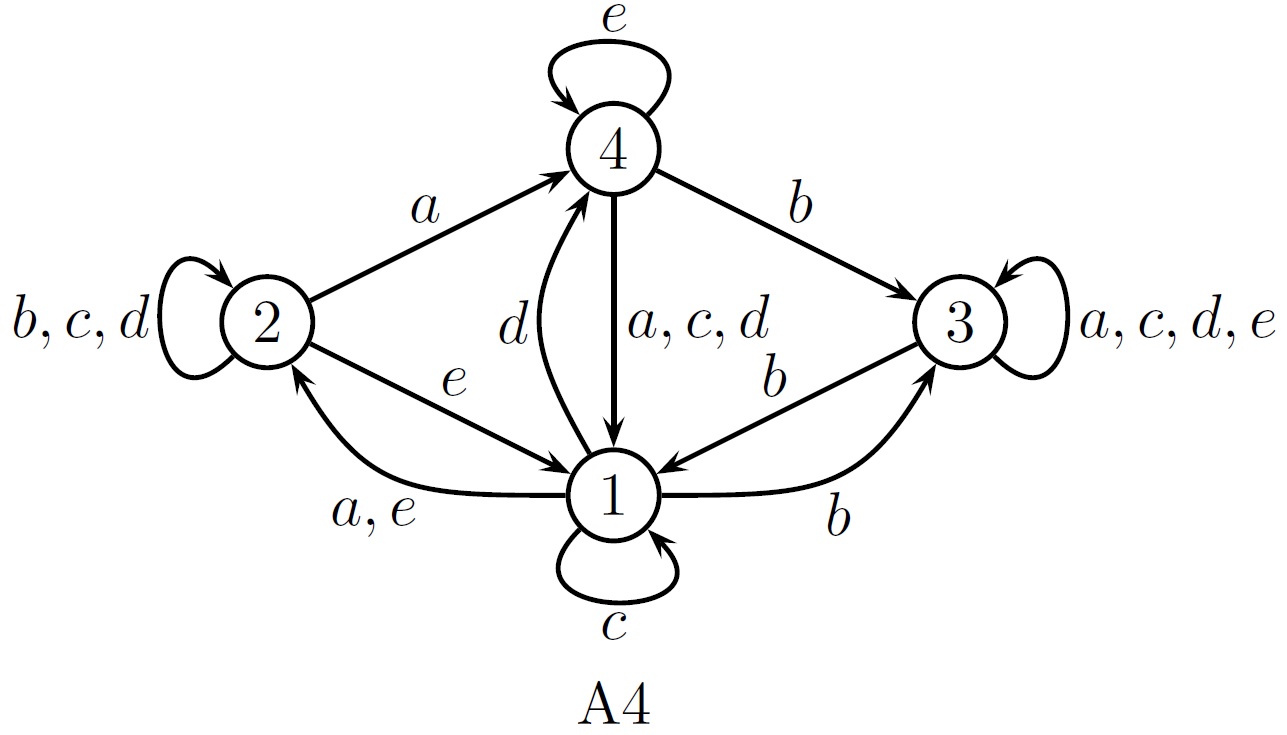}
\end{center}

Observe that A4 restricted to $a,b$ coincides with CPR and A4 restricted to $b,d,e$ coincides
with T4-1, so together with $C_4$ and T4-2 exactly the four automata with four states from
\citep{T06}, being the minimal ones.  On the other hand, $C_4$, T4-2 and A4 are the only maximal  basic critical DFAs on four states. We will prove this in Theorem \ref{thm4st}.
The DFA A4 admits 256 synchronizing words of length 9, expressed by the regular expression
$(b+c)(a+d)(a+e)(b+c)(a+e)b(a+d)(a+e)(b+c)$, where the synchronizing
state is 1 or 3, depending on the last symbol.
This follows from the analysis of the power set automaton of A4 that looks as
follows:

\begin{center}
\noindent\includegraphics[scale=0.28]{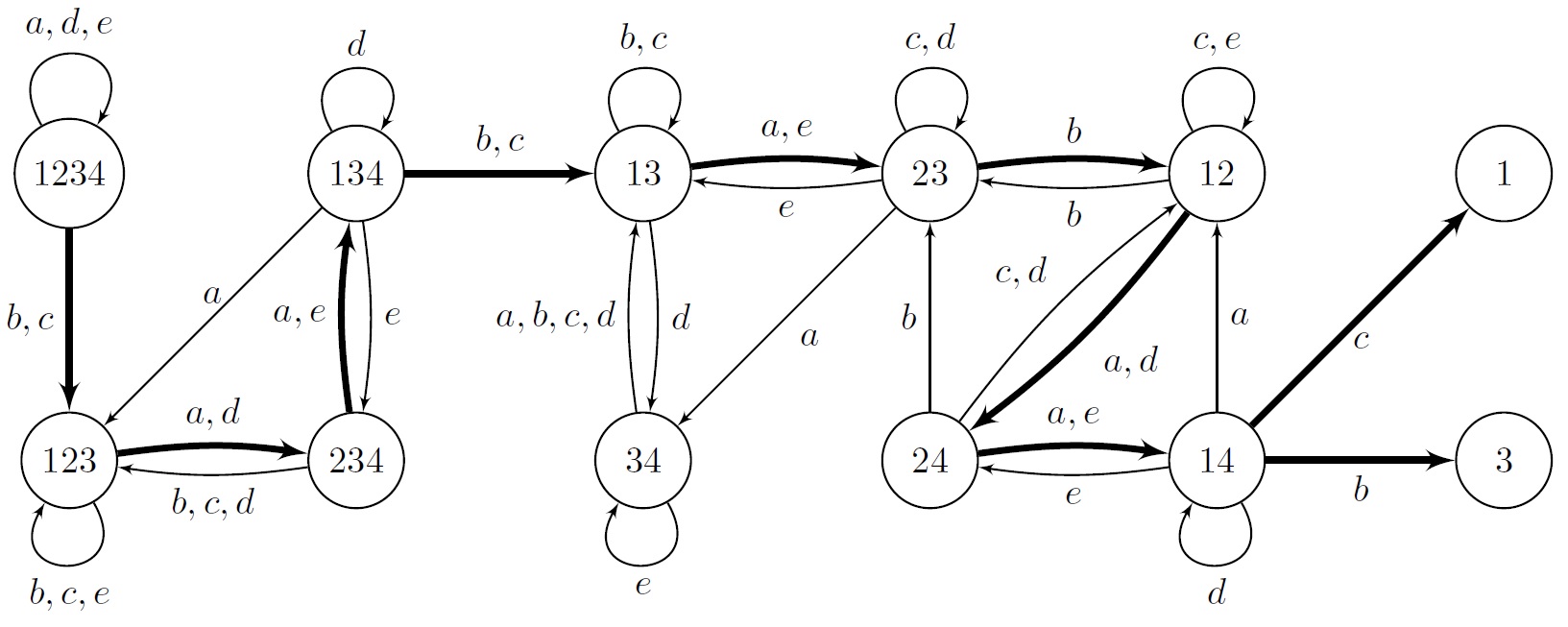}
\end{center}

Here the shortest paths from 1234 to a singleton are indicated by fat arrows.

The relationship between A4 and critical DFAs is given in the following theorem; it follows from the result of our algorithm.

\begin{theorem}
\label{thm4st}
No super-critical DFAs on four states exist, and a basic DFA on four states is critical if and
only if up to isomorphism it is $C_4$, T4-2, or one of the 10 automata that can be obtained
from A4 by removing zero or more symbols and keeping at least one of the sets
$\{a,b\}, \{b,d,e\}$ of symbols.
\end{theorem}

As a consequence of Theorem \ref{thm4st} apart from the four minimal critical DFAs that were
known on four states, we obtain 8 more that are not minimal.

\subsection{Five and Six States}

Here is the table for five states for minimal synchronization length $\geq 11$:

\begin{center}
\renewcommand{\arraystretch}{1.4}
\begin{tabular}{|r|rrrrrr|}
\hline
alph.\@ & sync.\@ & sync.\@ & sync.\@ & sync.\@ & sync.\@ & sync.\@ \\[-5pt]
size & 16 & 15 & 14 & 13 & 12 & 11 \\
\hline
1 & & & & & & \\[-5pt]
2 & 1 & 4 & 11 & 23 & 43 & 46 \\[-5pt]
3 & 1 & 19 & 85 & 280 & 1218 & 2580 \\[-5pt]
4 & & 36 & 275 & 1237 & 11310 & 36644 \\[-5pt]
5 & & 25 & 613 & 2837 & 57013 & 290466 \\[-5pt]
6 & & 5 & 915 & 4275 & 194115 & 1512125 \\[-5pt]
7 & & & 978 & 4799 & 497505 & 5658383 \\[-5pt]
8 & & & 774 & 4342 & 1011273 & 16136371 \\[-5pt]
9 & & & 454 & 3234 & 1672827 & 36527661 \\[-5pt]
10 & & & 194 & 1944 & 2284062 & 67619593 \\[-5pt]
11 & & & 58 & 912 & 2596207 & 104657920 \\[-5pt]
12 & & & 11 & 322 & 2468648 & 137653835 \\[-5pt]
13 & & & 1 & 81 & 1967657 & 155665867 \\[-5pt]
14 & & & & 13 & 1314222 & 152597099 \\[-5pt]
15 & & & & 1 & 733735 & 130410659 \\[-5pt]
16 & & & & & 340803 & 97538645 \\[-5pt]
17 & & & & & 130715 & 64001561 \\[-5pt]
18 & & & & & 40943 & 36877921 \\[-5pt]
19 & & & & & 10303 & 18643103 \\[-5pt]
20 & & & & & 2033 & 8241950 \\[-5pt]
21 & & & & & 303 & 3166721 \\[-5pt]
22 & & & & & 32 & 1047312 \\[-5pt]
23 & & & & & 2 & 294118 \\[-5pt]
24 & & & & & & 68851 \\[-5pt]
25 & & & & & & 13103 \\[-5pt]
26 & & & & & & 1957 \\[-5pt]
27 & & & & & & 219 \\[-5pt]
28 & & & & & & 17 \\[-5pt]
29 & & & & & & 1 \\
\hline
total & 2 & 89 & 4369 & 24300 & 15334969 & 1038664728 \\
\hline
\end{tabular}
\end{center}

\newpage
The table for six states  for minimal synchronization length $\geq 20$ is as follows.

\begin{center}
\renewcommand{\arraystretch}{1.4}
\begin{tabular}{|r|rrrrrr|}
\hline
alph.\@ & sync.\@ & sync.\@ & sync.\@ & sync.\@ & sync.\@ & sync.\@ \\[-5pt]
size & 25 & 24 & 23 & 22 & 21 & 20 \\
\hline
1 & & & & & & \\[-5pt]
2 & 2 & & 2 & 11 & 22 & 45 \\[-5pt]
3 & & & 4 & 63 & 282 & 718 \\[-5pt]
4 & & & & 158 & 1655 & 6596 \\[-5pt]
5 & & & & 267 & 5396 & 34248 \\[-5pt]
6 & & & & 324 & 11010 & 112063 \\[-5pt]
7 & & & & 271 & 15075 & 255313 \\[-5pt]
8 & & & & 152 & 14417 & 437701 \\[-5pt]
9 & & & & 54 & 9894 & 598072 \\[-5pt]
10 & & & & 11 & 4982 & 675859 \\[-5pt]
11 & & & & 1 & 1875 & 642233 \\[-5pt]
12 & & & & & 531 & 513958 \\[-5pt]
13 & & & & & 110 & 344369 \\[-5pt]
14 & & & & & 15 & 191471 \\[-5pt]
15 & & & & & 1 & 87307 \\[-5pt]
16 & & & & & & 32118 \\[-5pt]
17 & & & & & & 9310 \\[-5pt]
18 & & & & & & 2054 \\[-5pt]
19 & & & & & & 326 \\[-5pt]
20 & & & & & & 34 \\[-5pt]
21 & & & & & & 2 \\
\hline
total & 2 & 0 & 6 & 1312 & 65265 & 3943797 \\
\hline
\end{tabular}
\end{center}


So both for 5 and 6 states, up to symmetry there are exactly two basic critical DFAs.
Apart from $C_5$ and $C_6$ these are one on five states from
Roman \citep{R08} and one on six states from Kari \citep{K01}, depicted as follows.

\begin{center}
\includegraphics[scale=0.28]{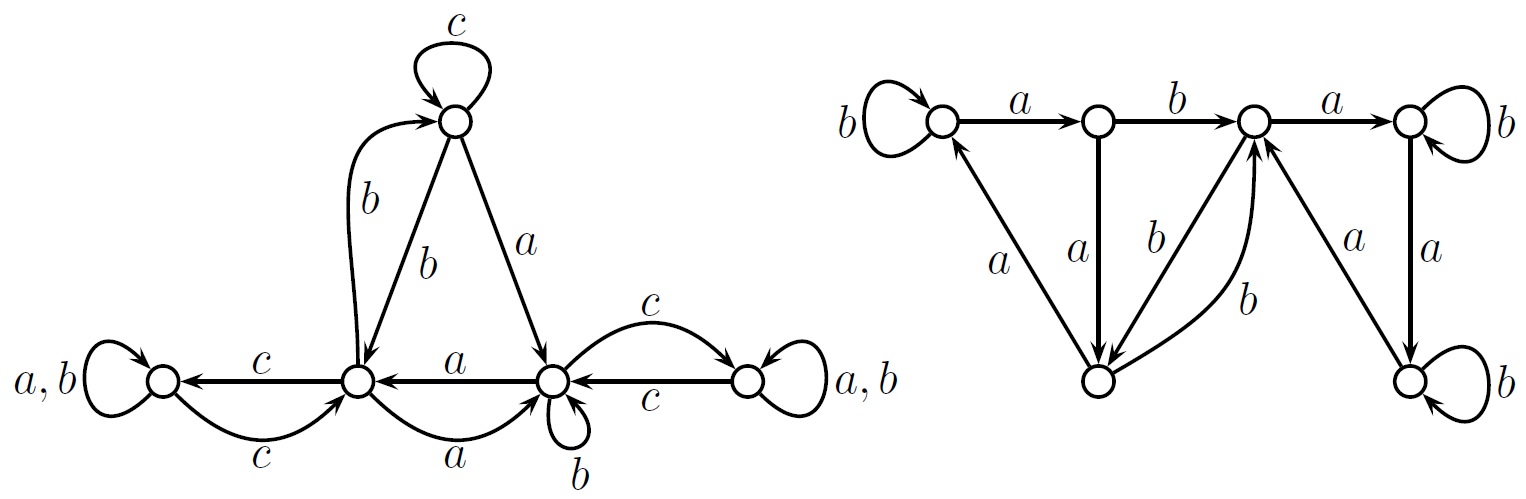}
\end{center}

For Roman's DFA the shortest synchronizing
word $abcacacbcaacabca$ is unique; for Kari's DFA there are two shortest synchronizing words,
described by \linebreak
$baabababaabbaba(baab+abaa)babaab$.

In contrast to 3 and 4 states, for 5 and 6 states there are no more basic critical DFAs than the minimal ones
that appeared in Trahtman's investigation.

In the table for 6 states we observe the first gap: DFAs with minimal synchronization length 25 exist (the two
critical DFAs), but no DFA with minimal synchronization length 24 exists, while for lengths $23, 22, 21,\ldots$
corresponding DFAs exist.

From these tables, we can extract the precise values of the maximal shortest synchronizing word length
$d(n,k)$ for DFA sizes $n = 2,3,4,5,6$, and alphabet sizes $k \le 17$ at least and $k \le 29$ at most,
depending on $n$. For practical reasons, the tables do not contain all our computational results.
For $n=4,5,6$, and larger alphabet sizes up to $41$, we have in addition that $d(n,k) = 5$, $d(n,k) = 10$, and
$d(n,k) = 19$ respectively. This leads to the following graph.

\begin{center}
\begin{tikzpicture}[x=2.25mm,y=2.25mm]
\begin{scope}
\clip (0,0) rectangle (40.7,25.7);
\foreach \x in {0,5,...,40} {
  \foreach \y in {0,10,20} {
    \fill[black!11] (\x,{\y+5*mod(\x,2)}) rectangle ++(5,5);
  }
}
\end{scope}
\foreach \x in {5,10,...,40} {
  \draw (\x,0) node[black,anchor=north] {$\scriptstyle\x$};
}
\foreach \y in {5,10,...,25} {
  \draw (0,\y) node[black,anchor=east] {$\scriptstyle\y\!$};
}
\draw (0,25.7) -- (0,-0.7) (40.7,0) -- (-0.7,0);
\mathversion{bold}
\draw[black!44,ultra thick,line cap=round,line join=round]
(1,1) -- (2,1) -- (3,1)
(1,2) -- (2,4) -- (5,4) -- (6,3) -- (9,3)
      -- (10,2) -- (23,2) -- (24,1) -- (26,1)
(1,3) -- (2,9) -- (5,9) -- (6,8) -- (8,8)
      -- (9,7) -- (17,7) -- (18,5)
      edge[line cap=butt] (40.7,5)
(1,4) -- (2,16) -- (3,16) -- (4,15) -- (6,15)
      -- (7,14) -- (13,14) -- (14,13) -- (15,13)
      -- (16,12) -- (23,12) -- (24,11) -- (29,11)
      -- (30,10) edge[line cap=butt] (40.7,10)
(1,5) -- (2,25) -- (3,23) -- (4,22) -- (11,22)
      -- (12,21) -- (15,21) -- (16,20) -- (21,20)
      -- (22,19) edge[line cap=butt] (40.7,19)
(3,1) node[anchor=west,scale=1.3] {$n=2$}
(26,1) node[anchor=west,scale=1.3] {$n=3$}
(40.7,5) node[anchor=west,scale=1.3] {$n=4$}
(40.7,10) node[anchor=west,scale=1.3] {$n=5$}
(40.7,19) node[anchor=west,scale=1.3] {$n=6$};
\mathversion{normal}
\draw (17.5,-2) node {$k$} edge[->] (23,-2)
(-3,12.5) node {$d(n,k)$} edge[->] (-3,18);
\end{tikzpicture}
\end{center}

From this graph, we see that the lower bounds from Theorem \ref{theorem:alphabet345}
for $k=3,4,5$ and $n = 3,4,5,6$ are only sharp for $k=n=3$.
For instance, $d(6,4) = 22$, where Theorem \ref{theorem:alphabet345} yields  $d(6,4) \geq 21$. This gives rise to the question whether the bounds for general $n$ can be improved.

Let's compare to the situation without alphabet size restrictions. For small values of $n$, quite some critical DFAs are known. However, as $n$ increases, such exceptional cases seem to evaporate rapidly and for general $n$ only \v{C}ern\'y's sequence is known. A similar phenomenon might be the case for fixed alphabet size as well: a pattern for general $n$ and some exceptional cases for small $n$.

There is at least some indication that the lower bound $d(n,3)\geq n^2-3n+4$ is an optimal general lower bound. It has been verified by Trahtman that there are no DFAs with three or four symbols for $n=7$ and synchronizing length exceeding $n^2-3n+4$. Also for $n> 7$, no DFAs are known with synchronizing length less than $(n-1)^2$ and strictly larger than $n^2-3n+4$. Furthermore, Trahtman's analysis confirms that no examples in this range exist for $n=8,9,10$ and alphabet size $2$.

\section{Extending $C_n$}
\label{seccn}

We observed that for $n=3,4$ there were non-maximal critical DFAs: DFAs that admit extensions that remain critical. For $n=5,6$ this did not occur. So it is a natural question
how this behaves for $n \geq 7$. In this section we show that then the DFA $C_n$, which is the only known critical DFA, is maximal: it cannot be extended
to a basic critical DFA.
The main result of this section is the following:

\begin{theorem}\label{theorem:extension}
Let $n\geq 5$ and let $C_{n}^c$ be a basic extension of $C_n$ by a symbol $c$. Then
$C_{n}^c$ admits a synchronizing word of length strictly
less than $(n-1)^2$.
\end{theorem}

Recall that {\em basic} means that $c$ is not equal to $a$ or $b$ and
that $c$ is not the identity function on $Q$. This section is
organized as follows: first we collect some properties of
$C_n$ and its unique shortest synchronizing word. Then
we consider the cases $|Qc|=n$, $|Qc| = n-1$ and $|Qc|\leq n-2$
separately.

\subsection{Properties of $C_n$}

Recall that $C_n$ is defined by
$n$ states $1,2,\ldots,n$, and two symbols $a,b$, acting by
$qa = q+1$ for $q = 1,\ldots,n-1$, $na = 1$, and  $qb = q$ for
$q = 2,\ldots,n$, $1b = 2$.
It is well known that $w_n = b(a^{n-1}b)^{n-2}$ of length  $|w_n| = (n-1)^2$ is its
shortest synchronizing word.  It is synchronizing since
\begin{eqnarray}
Qb &=& \left\{2,3,\ldots,n\right\}\\
\left\{2,3,\ldots,k\right\}a^{n-1}b &=&
\left\{2,3,\ldots,k-1\right\},\quad 3 \leq k\leq n.\label{eq:powerpath}
\end{eqnarray}
The first part of this word defines the path
\begin{equation}\label{eq:path} Q
\xrightarrow{b} Q\setminus\left\{1\right\}\xrightarrow{a}
Q\setminus\left\{2\right\}\xrightarrow{a}\ldots\xrightarrow{a}
Q\setminus\left\{n\right\}.
\end{equation}
We now extend the alphabet of the automaton by a non-trivial new
symbol $c$. Non-trivial means that the transitions defined by $c$
are not all equal to the transitions of $a$ or the transitions of
$b$ and furthermore that $c$ is not the identity function. We will
distinguish three cases:
\begin{enumerate}
\item $|Qc| = n$, i.e. $c$ is a permutation. \item $|Qc| = n-1$,
i.e. $c$ has deficiency 1. \item $|Qc| \leq n-2$, i.e. $c$ has
deficiency 2.
\end{enumerate}
We will show that in all these cases a shorter synchronizing word
exists. The general pattern in the arguments is as follows. The shortest synchronizing word $w_n$ corresponds to a path from $Q$ to a singleton in the power automaton of $C_n$. Take two sets $S,S'\subseteq Q$ on this path which are visited in this order. Let $d$ be the distance from $S$ to $S'$, i.e.
$$
d := \min\left\{|w|:Sw = S', w\in\left\{a,b\right\}^\star\right\}.
$$
Now construct a word $w\in\left\{a,b,c\right\}^\star$ in the automaton $C_n^c$ for which $Sw = S'$ and $|w|<d$. Then $C_n^c$ admits a synchronizing word of length at most $|w_n|-d+|w|<(n-1)^2$.

\subsection{Construction of a Shorter Synchronizing Word}

If $c$ defines a permutation on $Q$, we may assume that $c$
satisfies:
\begin{equation}\label{eq:permutationrequirement}
qc\leq q+1 \textrm{\ for\ all\ } q\in Q.
\end{equation}
Indeed, if $qc = q+k$ for some $q\in Q$ and $k\geq 2$, then
$(Q\setminus\left\{q\right\})c = Q\setminus\left\{q+k\right\}$,
which in view of (\ref{eq:path}) would imply existence of a
synchronizing word shorter than $(n-1)^2$. The following lemma
describes the structure of $c$.

\begin{lemma}\label{lemma:c_structure} If $|Q|=n\geq 1$ and $c$ is a permutation on $Q$ satisfying
(\ref{eq:permutationrequirement}), then there exist numbers $L$
(number of $c$-loops) and $1\leq l_1,\ldots,l_L \leq n$ (lengths
of $c$-loops) with $\sum_{i=1}^{L}l_i=n$ such that
\begin{equation}
qc = \left\{\begin{array}{lll}q-l_i+1&\textrm{if}\
q=l_1+\ldots + l_i\ \textrm{for\ some\ }1\leq i\leq L \\
q+1&\textrm{otherwise}\end{array}\right.
\end{equation}
\end{lemma}
An illustration of the statement is given below.

\includegraphics[scale=0.25]{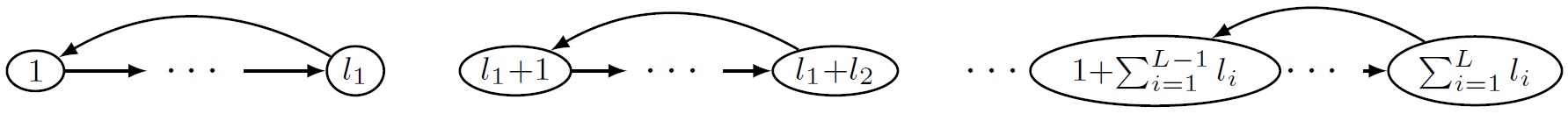}

\begin{proof} We give a proof by induction. For $n=1$,
$1\xrightarrow{c}1$, so $L=1$ and $l_1=1$.
Now suppose the statement is true for all $n\leq N$ and consider
the case $|Q|=N+1$. If $1\xrightarrow{c}1$, then $c$ defines a
permutation on $Q\setminus\left\{1\right\}$. Applying the
induction hypothesis on $Q\setminus\left\{1\right\}$ gives the
result. If
$1\xrightarrow{c}2\xrightarrow{c}\ldots\xrightarrow{c}k$ for some
$k\geq 2$, then either $kc = k+1$ or $kc=1$. In both cases there
is a number $l_1\geq 1$ such that $1\xrightarrow{c}\ldots
\xrightarrow{c}l_1\xrightarrow{c}1$. Apply the induction
hypothesis on the remaining $n-l_1$ states.\qed
\end{proof}

Note that $L=1$ and $L=n$ are the trivial cases, because then
$c=a$ or $c$ is the identity. Before we give a general argument,
we first give an example.

\includegraphics[scale=0.25]{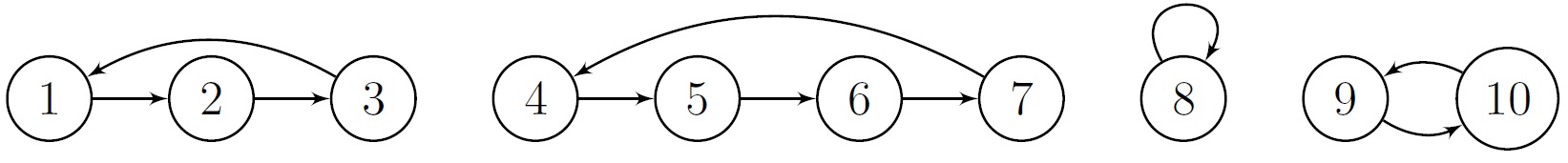}

\begin{example}
Consider the automaton $C_{10}^c = \left\{Q,\Sigma,\delta\right\}$
 with $Q = \left\{1,\ldots,10\right\}$ and $\Sigma =
\left\{a,b,c\right\}$. The actions of the symbols $a$ and $b$ are
from the definition of $C_n$ and $c$ is the permutation shown above.
Here we have four loops ($L=4$) with lengths $l_1 = 3, l_2 = 4, l_3 = 1$ and $l_4 = 2$.
We will show how to use the
$c$-loop of length four to create a shorter synchronizing word.
Consider the set $S = \left\{2,\ldots,9\right\}$.
We start by a greedy approach to reach a set of size 7:
$$
Sa^3b = \left(\left\{1,2\right\}\cup \left\{5,\ldots,10\right\}\right)b = \left\{2\right\}\cup \left\{5,\ldots,10\right\}.
$$
As a next step, we shift everything by using the symbol $a$ until
the isolated state $\left\{2\right\}$ ends up in the $c$-loop of
length four:
$$
\left(\left\{2\right\}\cup \left\{5,\ldots,10\right\}\right)a^3 = \left\{1,2,3\right\}\cup\left\{5\right\}\cup\left\{8,9,10\right\}
$$
Since $\left\{1,2,3\right\}$ and $\left\{8,9,10\right\}$ are
(unions of) full $c$-loops, they are invariant under $c$.
Therefore, we can move the isolated state $\left\{5\right\}$ to
the desired position:
$$
\left(\left\{1,2,3\right\}\cup\left\{5\right\}\cup\left\{8,9,10\right\}\right)c^3 = \left\{1,2,3,4\right\}\cup\left\{8,9,10\right\}
$$
Finally, we shift again by a power of $a$ and apply $b$ to get rid
of one more state:
$$
\left(\left\{1,2,3,4\right\}\cup\left\{8,9,10\right\}\right)a^3b = \left\{1,\ldots,7\right\}b = \left\{2,\ldots,7\right\}:=S'.
$$
We conclude that the word $w = a^3ba^3c^3a^3b$ has the property
that $Sw = S'$. In $C_{10}$ both $S$ and $S'$ are on the
shortest path from $Q$ to $\left\{2\right\}$ and by
(\ref{eq:powerpath}) the distance between them is equal to
$2n=20$. The word $w$ has length $|w| = 14$, so in
$C_{10}^c$ there exists a synchronizing word of length
at most $(10-1)^2-6 = 75$. Note that there might be even shorter
synchronizing words, but for our main goal it is sufficient to
have some synchronizing word shorter than $81$.
\end{example}

The idea of this example works in more generality if there is a
$c$-loop of length at least 3, as is proved in the next lemma. If
the longest loop has length $2$, then basically we can do the same
thing, but we need at least three $c$-loops to isolate a state.

\begin{lemma}\label{lemma:permutation}
Let $n\geq 5$ and let $C_{n}^c$ be an extension of the automaton
$C_n$ by a symbol $c$ as given in Lemma
\ref{lemma:c_structure}. If $2\leq L \leq n-1$, then
$C_{n}^c$ admits a synchronizing word of length strictly
less than $(n-1)^2$.
\end{lemma}
\begin{proof} We distinguish the following three cases:
\begin{itemize}
\item $L\geq 2$ and $l_k\geq 3$ for some $k$.

\item $L\geq 3$ and $l_k= 2$ for some $k\leq L-1$.

\item $L\geq 3$ and $l_L = 2$.
\end{itemize}
Note that for all $n\geq 5$ and all possible non-trivial choices
of $c$, the extended automaton $C_n^+$ satisfies at
least one of these cases.

\paragraph{Case 1: $L\geq 2$ and $l_k\geq 3$ for some $k$.}
Take $k$ such that $l_k\geq 3$ and write
$\Lambda^- = \sum_{i=1}^{k-1}l_i$, $\Lambda^+ =\sum_{i=k+1}^Ll_i$,
for the sum of the loop lengths before the $k$th loop and after the
$k$th loop respectively. These sums can be zero if $k=1$ or $k=L$.
Define $\Lambda = \Lambda^-+\Lambda^+=n-l_k\leq n-3$. Since $L\geq
2$, we have $\Lambda\geq 1$. Take
$$
S = \left\{2,3,\ldots,n-l_k+3\right\},\quad S' = \left\{2,3,\ldots,n-l_k+1\right\}.
$$
and define the word
\begin{equation}\label{eq:word_case1}
w = a^{l_k-1}ba^{\Lambda^-}c^{l_k-1}a^{\Lambda^+}b.
\end{equation}
We will show that $Sw=S'$. Write $S=S_1\cup S_2$ with
\begin{eqnarray*}
S_1 &=& \left\{2,\ldots,n-l_k+1\right\} = \left\{2,\ldots,1+\Lambda\right\},\\
S_2 &=& \left\{n-l_k+2,n-l_k+3\right\} = \left\{2+\Lambda,3+\Lambda\right\}.
\end{eqnarray*}
Then
\begin{eqnarray}
S_1w &=& \left\{2,\ldots,1+\Lambda\right\}a^{l_k-1}ba^{\Lambda^-}c^{l_k-1}a^{\Lambda^+}b\nonumber\\
&=&
\left\{l_k+1,\ldots,n\right\}ba^{\Lambda^-}c^{l_k-1}a^{\Lambda^+}b\nonumber\\
&=&
\left\{l_k+1,\ldots,n\right\}a^{\Lambda^-}c^{l_k-1}a^{\Lambda^+}b\nonumber\\
&=& \left(\left\{1,\ldots,\Lambda^-\right\}\cup\left\{\Lambda^-+l_k+1,\ldots,n\right\}\right)c^{l_k-1}a^{\Lambda^+}b\label{eq:c_occurs}\\
&=&
\left(\left\{1,\ldots,\Lambda^-\right\}\cup\left\{\Lambda^-+l_k+1,\ldots,n\right\}\right)a^{\Lambda^+}b\nonumber\\
&=& \left\{1,\ldots,\Lambda\right\}b\nonumber\\ &=&
\left\{\begin{array}{lll}\left\{2\right\}=\left\{1+\Lambda\right\}&\textrm{if}&\Lambda=1\\\left\{2,\ldots,\Lambda\right\}
&\textrm{if}&\Lambda\geq 2,\end{array}\right.\nonumber
\end{eqnarray}
where sets of the form $\left\{x,\ldots,y\right\}$ with $x>y$
should be interpreted as being empty. This occurs if $\Lambda^-=0$
or $\Lambda^+=0$. Furthermore
\begin{eqnarray}
S_2w &=& \left\{2+\Lambda,3+\Lambda\right\}a^{l_k-1}ba^{\Lambda^-}c^{l_k-1}a^{\Lambda^+}b = \left\{1,2\right\}ba^{\Lambda^-}c^{l_k-1}a^{\Lambda^+}b\nonumber\\ &=& \left\{2\right\}a^{\Lambda^-}c^{l_k-1}a^{\Lambda^+}b = \left\{2+\Lambda^-\right\}c^{l_k-1}a^{\Lambda^+}b= \left\{1+\Lambda^-\right\}a^{\Lambda^+}b\label{eq:c_occurs2}\\
&=& \left\{1+\Lambda\right\}b= \left\{1+\Lambda\right\}.\nonumber
\end{eqnarray}
It follows that the word $w$ has the property
$$
Sw = (S_1\cup S_2)w = S_1w\cup S_2w = \left\{2,\ldots,\Lambda+1\right\} = S'.
$$
and its length is
$|w| = l_k-1+1+\Lambda^-+l_k-1+\Lambda^++1 = 2l_k+\Lambda = l_k+n <2n$.
In the automaton $C_n$ the sets $S$ and $S'$ are both on
the shortest path from $Q$ to a singleton and the shortest path is
defined by $S(a^{n-1}b)^2 = S'$. Since $|(a^{n-1}b)^2| = 2n >
|w|$, the statement of the lemma follows.

The above proof fails in case $l_k \leq 2$, since then $n-l_k+3
>n$. However, the proofs for the other cases use pretty much the same ideas.

\paragraph{Case 2: $L\geq 3$ and $l_k= 2$ for some $k\leq L-1$.}
Take $k$ such that $l_k= 2$ and write
$$
\Lambda^- = \sum_{i=1}^{k-1}l_i,\quad \Lambda^+
=\sum_{i=k+2}^Ll_i,
$$
for the sum of the loop lengths before the $k$th loop and after the
$(k+1)$th loop respectively. These sums can be zero if $k=1$ or
$k=L-1$. Define $\Lambda = \Lambda^-+\Lambda^+ = n-l_k-l_{k+1}\leq
n-3$. From the assumption $L\geq 3$ it follows that $\Lambda\geq
1$. Take
$$
S = \left\{2,3,\ldots,\Lambda+3\right\},\quad S' =
\left\{2,3,\ldots,\Lambda+1\right\}.
$$
and define the word
$$
w = a^{l_k+l_{k+1}-1}ba^{\Lambda^-}ca^{\Lambda^+}b.
$$
By a similar argument as in Case 1 it follows that $Sw=S'$: Let
$S_1 = \left\{2,\ldots,\Lambda+1\right\}$, then
\begin{eqnarray}
S_1w &=&
\left\{2,\ldots,\Lambda+1\right\}a^{l_k+l_{k+1}-1}ba^{\Lambda^-}ca^{\Lambda^+}b\nonumber\\
&=& \left\{l_k+l_{k+1}+1,\ldots,n\right\}ba^{\Lambda^-}ca^{\Lambda^+}b\nonumber\\
&=& \left\{l_k+l_{k+1}+1,\ldots,n\right\}a^{\Lambda^-}ca^{\Lambda^+}b\nonumber\\
&=& \left(\left\{1,\ldots,\Lambda^-\right\}\cup\left\{\Lambda^-+l_k+l_{k+1}+1,\ldots,n\right\}\right)ca^{\Lambda^+}b\label{eq:c_occurs3}\\
&=& \left(\left\{1,\ldots,\Lambda^-\right\}\cup\left\{\Lambda^-+l_k+l_{k+1}+1,\ldots,n\right\}\right)a^{\Lambda^+}b\nonumber\\
&=& \left\{1,\ldots,\Lambda\right\}b =
\left\{\begin{array}{lll}\left\{2\right\}=\left\{1+\Lambda\right\}&\textrm{if}&\Lambda=1\\\left\{2,\ldots,\Lambda\right\}
&\textrm{if}&\Lambda\geq 2,\end{array}\right.\nonumber
\end{eqnarray}
Completely analogous to Case 1, we have
$$
\left\{\Lambda+2,\Lambda+3\right\}w = \left\{1+\Lambda\right\}.
$$
Therefore,
$$
Sw =
\left\{2,\ldots,\Lambda+1\right\}w\cup\left\{\Lambda+2,\Lambda+3\right\}w
= \left\{2,\ldots,\Lambda+1\right\} = S'
$$
Since $w$ has length $n+2<2n$, the statement of the lemma follows.

\paragraph{Case 3: $L\geq 3$ and $l_L = 2$.}
Define
\begin{equation}\label{eq:word_case3}
S = \left\{2,\ldots,n\right\},\qquad w = a^2ba^{n-3}cab.
\end{equation}
Then
\begin{eqnarray}
Sw &=&
(\left\{1,2\right\}\cup\left\{4,\ldots,n\right\})ba^{n-3}cab= (\left\{2\right\}\cup\left\{4,\ldots,n\right\})a^{n-3}cab\nonumber\\
&=& (\left\{n-1\right\}\cup\left\{1,\ldots,n-3\right\})cab= (\left\{n\right\}\cup\left\{1,\ldots,n-3\right\})ab\label{eq:c_occurs4}\\
&=& \left\{1,\ldots,n-2\right\}b =
\left\{2,\ldots,n-2\right\}.\nonumber
\end{eqnarray}
Since $|w|=n+3<2n$, the result follows.\qed
\end{proof}

\subsection{The Additional Symbol has Deficiency 1}

In this section we assume that the additional symbol $c$ satisfies
$|Qc|=n-1$. We will prove that the extended automaton
$C_n^c$ admits a synchronizing word of length strictly
less than $(n-1)^2$ for every non-trivial choice of $c$. The first
step (Lemma's \ref{lemma:case2a}, \ref{lemma6}, \ref{lemma7} and
Corollary \ref{cor8}) is to show that the only candidates to
preserve the shortest synchronizing word length have a loop
structure similar to the permutations in Lemma
\ref{lemma:c_structure}. In Lemma \ref{lemma9} we couple such
candidates $c$ to a permutation $\tilde c$, which leads to the
conclusion that the automaton with $c$ synchronizes at least as
fast as the automaton with $\tilde c$.

\begin{lemma}\label{lemma:case2a}
Let $n\geq 5$ and let $C_n^c$ be an extension of the automaton
$C_n$ by a symbol $c$ for which $|Qc|=n-1$. If
the shortest synchronizing word for $C_n^c$ has length
$(n-1)^2$, then $Qc = Q\setminus\left\{1\right\}$ and $c$ defines
a permutation on $Q\setminus\left\{1\right\}$.
\end{lemma}

\begin{proof} If $Qc = Q\setminus\left\{q\right\}$ with $q\neq
1$, then $w = ca^{n-q}b(a^{n-1}b)^{n-3}$ is synchronizing and $w$
has length
$$
|w| = 1+n-q+1+n(n-3) = (n-1)^2-q+1 <(n-1)^2.
$$
If $Qc=Q\setminus\left\{1\right\}$ and $|Qc^2|\leq n-2$, then one
of the following two is true:
\begin{itemize}
\item $Qc^2 =Q\setminus\left\{1,2\right\}$.\\
In this case $w = c^2a^{n-2}b(a^{n-1}b)^{n-4}$ is synchronizing
and has length
$$
|w| = 2+n-2+1+n(n-4) = (n-1)^2-n<(n-1)^2.
$$
\item $Qc^2\subset Q\setminus\left\{q\right\}$ for some $q\geq
3$.\\
In this case $w = c^2a^{n-q}b(a^{n-1}b)^{n-3}$ is synchronizing
and $w$ has length
$$
|w| = 2+n-q+1+n(n-3) = (n-1)^2-q+2 <(n-1)^2.
$$
\end{itemize}
Therefore, we may assume that $Qc=Q\setminus\left\{1\right\}$ and
$|Qc^2|= n-1$. This means that
$(Q\setminus\left\{1\right\})c=Q\setminus\left\{1\right\}$, so $c$
defines a permutation on $Q\setminus\left\{1\right\}$.\qed\end{proof}

The next lemma shows that $c$ can be assumed to satisfy $qc\leq
q+1$ for all $q$.

\begin{lemma}\label{lemma6}
Let $n\geq 5$ and let $C_n^c$ be an extension of the automaton
$C_n$ by a symbol $c$ for which $|Qc|=n-1$. If
$qc = q+k$ for some $q\in Q$ and $k\geq 2$, then $C_n^c$
admits a synchronizing word of length strictly less than
$(n-1)^2$.
\end{lemma}
\begin{proof} If $qc = q+k$ for some $q\in Q$ and $k\geq 2$,
then either $1c \neq 2$ or $qc = q+k$ for some $q\geq 2$ and
$k\geq 2$. We distinguish these two cases:
\begin{itemize}
\item $1c \neq 2$. In this case there exists a singleton $\tilde q
:= 2c^{-1}$, so
$$
(Q\setminus\left\{\tilde q\right\})c =
Q\setminus\left\{1,2\right\}.
$$
The sets $Q\setminus\left\{\tilde q\right\}$ and
$Q\setminus\left\{1,2\right\}$ are both on the shortest path in
$C_n$, where
$$
(Q\setminus\left\{\tilde q\right\})a^{n-\tilde q}ba =
Q\setminus\left\{1,2\right\}.
$$
Since $a^{n-\tilde q}ba\geq 2$, the shortest synchronizing word in
$C_n^c$ has length at most $(n-1)^2-1$.

\item $1c = 2$ and there exist $q\geq 2$ and $k\geq 2$ such that
$qc=q+k$. In this case
$$
(Q\setminus\left\{q\right\})c = Q\setminus\left\{1,q+k\right\}
\subseteq Q\setminus\left\{q+k\right\},
$$
which means that there is synchronizing word of length
$(n-1)^2-k+1$ in $C_n^c$, see (\ref{eq:path}).\qed
\end{itemize}\end{proof}

\begin{lemma}\label{lemma7}
Suppose $|Q|=n\geq 2$ and $c$ is such that
\begin{equation}
Qc = Q\setminus\left\{1\right\},\quad
\left(Q\setminus\left\{1\right\}\right)c =
Q\setminus\left\{1\right\}\quad\textrm{and}\quad qc \leq q+1
\textrm{\ for\ all\ } q.
\end{equation}
Then there exist numbers $L$ (number of $c$-loops) and $1\leq
l_1,\ldots,l_L \leq n-1$ (lengths of $c$-loops) with
$\sum_{i=1}^{L}l_i=n-1$ such that
\begin{eqnarray}
qc &=& \left\{\begin{array}{ll}q-l_i+1&\textrm{if}\
q=l_1+\ldots + l_i+1 \ \textrm{for\ some\ }1\leq i\leq L \\
q+1&\textrm{otherwise}\end{array}\right.
\end{eqnarray}
\end{lemma}

\begin{proof} Similar to the proof of Lemma
\ref{lemma:c_structure}.\qed\end{proof}

\begin{corollary}\label{cor8}
Let $n\geq 5$ and let $C_n^c$ be an extension of the automaton
$C_n$ by a symbol $c$ for which $|Qc|=n-1$. If
the shortest synchronizing word for $C_n^c$ has length
$(n-1)^2$, then $c$ has the structure described in Lemma
\ref{lemma7}.
\end{corollary}

An illustration of the statement is given below.
The structure of $c$ if $|Qc|=n-1$.
Dotted arrows represent chains of transitions of the form $qc = q+1$.

\begin{center}

\includegraphics[scale=0.25]{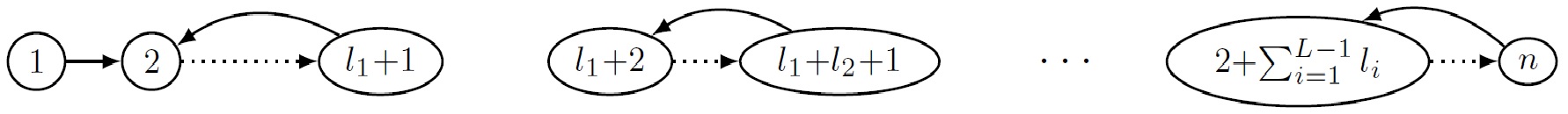}

\end{center}

Finally, in the next lemma, we handle symbols $c$ having the
structure described in Lemma \ref{lemma7}. If all loops of $c$
have length 1, then $qc = qb$ for all $q\in Q$. Therefore the case
$L=n-1$ is excluded.

\begin{lemma}\label{lemma9}
Let $n\geq 5$ and let $C_{n}^c$ be an extension of the automaton
$C_n$ by a symbol $c$ as given in Lemma
\ref{lemma7}. If $1\leq L \leq n-2$, then $C_{n}^c$
admits a synchronizing word of length strictly less than
$(n-1)^2$.
\end{lemma}

\begin{proof} We distinguish two cases: $2\leq l_1\leq n-1$ and
$l_1=1$.\\ \textbf{Case 1: $2\leq l_1\leq n-1$}. In this case
$$
Q\setminus\left\{l_1\right\} \xrightarrow{c}
Q\setminus\left\{1,l_1+1\right\}\xrightarrow{c}
Q\setminus\left\{1,2\right\}.
$$
In $C_n$ the shortest path between these sets is given
by
$$
Q\setminus\left\{l_1\right\}\xrightarrow{a^{n-l_1}}
Q\setminus\left\{n\right\} \xrightarrow{b}
Q\setminus\left\{1,n\right\} \xrightarrow{a}
Q\setminus\left\{1,2\right\},
$$
which has length $n-l_1+2\geq 3$. Therefore, $C_n^c$ has
a synchronizing word of length at most $(n-1)^2-1$.

\paragraph{Case 2: $l_1 = 1$.} This means that $2c = 2$. We define a
permutation $\tilde c$ on $Q$ by
$$
q\tilde c = \left\{\begin{array}{ll}1&\textrm{if\
}q=1,\\qc&\textrm{if\ }q\neq 1.\end{array}\right.
$$
The permutation $\tilde c$ has $\tilde L:= L+1\geq 2$ loops and
$L$ of them coincide with the loops of $c$. Since $c$ has a loop
of length at least 2, so does $\tilde c$. The loop lengths of
$\tilde c$ are given by $\tilde l_1 = 1$ and $\tilde l_{k} =
l_{k-1}$ for $2\leq k\leq \tilde L$

By Lemma \ref{lemma:permutation} we already know that there exists
a synchronizing word $\tilde w\in\left\{a,b,\tilde c\right\}^*$
with $|\tilde w|<(n-1)^2$. Define $w\in\left\{a,b,c\right\}^*$ as
the word that is obtained from $\tilde w$ by replacing all
instances of $\tilde c$ by $c$. Clearly this operation preserves
the word length. We will show that the word $w$ is a synchronizing
word for $C_n^c$.

The key observation is that the permutation $\tilde c$ has the
following property for $S\subseteq Q$:
\begin{equation}\label{eq:property}
\textrm{If $1\not\in S$ or $2\in S$, then $Sc^k\subseteq S\tilde c
^k$ for all $k\geq 1$}.
\end{equation}
We consider the same cases as in the proof of Lemma
\ref{lemma:permutation}:
\begin{itemize}
\item $\tilde L \geq 2$ and $\tilde l_k\geq 3$ for some $k$. In
this case
$$
\tilde w = a^{\tilde l_k-1}ba^{\Lambda^-}\tilde c^{\tilde
l_k-1}a^{\Lambda^+}b
$$
is synchronizing (compare to (\ref{eq:word_case1})), where
$$
\Lambda^- =\sum_{i=1}^{k-1}\tilde l_i \geq \tilde l_1+\tilde l_2 =
2,\qquad \Lambda^+ = \sum_{i=k+1}^{\tilde L}\tilde l_i.
$$
Here we used that $\tilde l_1 = \tilde l_2 = 1$ and $k\geq 3$
since $1\tilde c = 1$ and $2\tilde c = 2$. Let
$$
T_1 =
\left\{1,\ldots,\Lambda^-\right\}\cup\left\{\Lambda^-+l_k+1,\ldots,n\right\}
\quad\textrm{and}\quad T_2 = \left\{2+\Lambda^-\right\},
$$
and observe that $2\in T_1$ and $1\not\in T_2$. By property
(\ref{eq:property}), we obtain
$$
T_1 c^{\tilde l_k-1} \subseteq  T_1 \tilde c^{\tilde l_k-1} ,\quad
T_2 c^{\tilde l_k-1} \subseteq  T_2 \tilde c^{\tilde l_k-1}.
$$
Comparing with the argument in the proof of Lemma
\ref{lemma:permutation}, in particular (\ref{eq:c_occurs}) and
(\ref{eq:c_occurs2}), we conclude that $Qw\subseteq Q\tilde w$ and
$w$ is synchronizing.

\item $\tilde L \geq 3$ and $\tilde l_k=2$ for some $k\leq \tilde
L-1$. Here an analogous argument as in the previous case gives the
result.

\item $\tilde L \geq 3$ and $\tilde l_{\tilde L} =2$. Let
\begin{equation*}
\tilde w = a^2ba^{n-3}\tilde cab,
\end{equation*}
analogous to (\ref{eq:word_case3}). Since $n\geq 5$, we have
$$
2\in \left\{n-1\right\}\cup \left\{1,\ldots,n-3\right\}.
$$
Applying property (\ref{eq:property}) again, we obtain
$$
\left(\left\{n-1\right\}\cup \left\{1,\ldots,n-3\right\}\right)c
\subseteq \left(\left\{n-1\right\}\cup
\left\{1,\ldots,n-3\right\}\right)\tilde c.
$$
By comparing with (\ref{eq:c_occurs4}), it follows that
$Qw\subseteq Q\tilde w$ and therefore $w$ synchronizes.\qed
\end{itemize}\end{proof}

\subsection{The Additional Symbol has Deficiency at least 2}

\begin{lemma}
Let $n\geq 5$ and let $C_{n}^c$ be an extension of the automaton
$C_n$ by a symbol $c$ such that $|Qc|\leq n-2$.
Then $C_{n}^c$ admits a synchronizing word of length
strictly less than $(n-1)^2$.
\end{lemma}

\begin{proof} There exists $q\geq 2$ such that $Qc \subset
Q\setminus\left\{q\right\}$, which implies the
result.\qed\end{proof}

\noindent\textit{Proof of Theorem \ref{theorem:extension}.} Combining all
results of the preceding sections completes the proof.\qed

\section{Conclusions and Further Research}
\label{secconc}
We investigated slowly synchronizing DFAs in two main ways. The first one is exploiting computer support for a full investigation
of such DFAs on $n$ states for $n\leq 6$. The second way is proving properties in classical mathematical style: we developed
lower bounds on synchronization length not only depending on DFA size $n$ but also on the alphabet size, and we proved that $C_n$ does not
admit non-trivial critical extensions. As remarkable results we mention:
\begin{itemize}
\item Synchronization lengths close to $(n-1)^2$ can be obtained for large (even exponential) alphabet size.
\item In contrast to what Trahtman expected, several minimal critical DFAs on 3 and 4
states can be combined and/or extended to new critical DFAs. For all of these the minimal synchronizing
word is not unique, and sometimes the synchronizing state is not unique.
\end{itemize}

Despite of extensive effort,  \v{C}ern\'y's conjecture is still open after more than half a
century.
Being a strengthening of this long standing open problem, a full characterization of all critical
DFAs (expected to only consist of $C_n$ and the critical DFAs on $\leq 6$ states investigated in this paper) may not be
tractable. More feasible challenges may include
\begin{itemize}
\item proving or disproving that every non-minimal basic critical DFA admits multiple shortest synchronizing words,
\item giving an upper bound on the number of symbols in a minimal
critical DFA (all known examples have at most three),
\item improve bounds on $d(n,k)$ or prove they are tight,
\item proving or disproving that $d(n,k+1) \leq d(n,k)$ for all $k\geq 2$.
\end{itemize}

\section*{References}

\bibliography{ref}
\bibliographystyle{abbrv}

\end{document}